\documentclass[aps,prx,twocolumn,superscriptaddress,floatfix]{revtex4-2}
\usepackage{amssymb}
\usepackage{amsmath}
\usepackage[dvipsnames]{xcolor}
\usepackage{color}         % produces boxes or entire pages with colored backgrounds
\usepackage{graphics}      % standard graphics specifications
\usepackage[pdftex]{graphicx}      % alternative graphics specifications
\usepackage{epsf}          % old package handles encapsulated post script issues
\usepackage{bm}            % special 'bold-math' package
\usepackage{asymptote}     % For typesetting of mathematical illustrations
\usepackage{thumbpdf}
\usepackage{verbatim}
\usepackage{url}
\usepackage{MnSymbol}
\usepackage{physics}
\usepackage{graphicx}
\usepackage[colorlinks=true, hyperindex, breaklinks, linkcolor=blue, urlcolor=blue, citecolor=blue]{hyperref} % Physical Review style
\usepackage{enumitem}	% compact enumeration

\usepackage{float}
\usepackage{tipa, extraipa}
\usepackage{color}
\usepackage{verbatim}
\usepackage{soul,xcolor}
\setstcolor{red}
\usepackage{colortbl}
\usepackage{amsthm}
\usepackage[linesnumbered, ruled, vlined]{algorithm2e}
\usepackage{makecell}
\usepackage{stmaryrd}
\usepackage{tikz-cd}
\usepackage{dsfont}
\usepackage{array} % Put this in the preamble
\usepackage{blkarray}
\usepackage{booktabs}
\usepackage{subfigure}

\newtheorem{theorem}{Theorem}
\newtheorem{lemma}[theorem]{Lemma}
\newtheorem{corollary}[theorem]{Corollary}
\newtheorem{definition}[theorem]{Definition}
\newtheorem{proposition}[theorem]{Proposition}

\newcommand{\mc}{\mathcal}
\newcommand{\mb}{\mathbf}
\newcommand{\mbb}{\mathbb}
\newcommand{\mr}{\mathrm}

\renewcommand{\ker}[1]{\mathrm{ker}(#1)}
\newcommand{\im}[1]{\mathrm{im}(#1)}

\newcommand{\QAB}{\mathcal{Q}_{\mathrm{AB}}}
\newcommand{\QBC}{\mathcal{Q}_{\mathrm{BC}}}
\newcommand{\QAC}{\mathcal{Q}_{\mathrm{AC}}}
\newcommand{\QABC}{\mathcal{Q}_{\mathrm{ABC}}}

\usepackage{tabularx}

\newcommand{\tQAB}{\tilde{\mathcal{Q}}_{\mathrm{AB}}}

\usepackage[normalem]{ulem}

\SetKwInOut{Input}{Input}
\SetKwInOut{Output}{Output}
\SetKwFunction{FMain}{Rearrange}
\SetKwFunction{Fappend}{append}
\SetKwProg{Fn}{Function}{:}{}
\RestyleAlgo{ruled}

\begin{document}
\title{
Transversal Dimension Jump for Product qLDPC Codes
}

\author{Christine Li}
\affiliation{Institute for Quantum Information and Matter, Caltech, Pasadena, CA 91125, USA}
\affiliation{Fu Foundation School of Engineering and Applied Science, Columbia University, New York, NY 10027, USA
}

\author{John Preskill}
\email{preskill@caltech.edu}
\affiliation{Institute for Quantum Information and Matter, Caltech, Pasadena, CA 91125, USA}
\affiliation{AWS Center for Quantum Computing, Pasadena, CA 91125, USA}

\author{Qian Xu}
\email{qianxu@caltech.edu}
\affiliation{Institute for Quantum Information and Matter, Caltech, Pasadena, CA 91125, USA}
\affiliation{Walter Burke Institute for Theoretical Physics, Caltech, Pasadena, CA 91125, USA}

\begin{abstract}
We introduce \emph{transversal dimension jump}, a code-switching protocol for lifted product (LP) quantum low-density parity-check (qLDPC) codes across different chain-complex dimensions, enabling universal fault-tolerant quantum computation with low overhead. The construction leverages the product structure of LP codes to implement one-way transversal CNOTs between a 3D code and its 2D component codes, enabling teleportation-based switching with geometrically nonlocal gates. 
Combined with constant-depth CCZ gates in 3D LP codes and low-overhead transversal Clifford gates in 2D LP codes, this yields universal, high-rate quantum logical computation with high thresholds and low space-time costs.
Beyond asymptotic schemes, we identify explicit 3D–2D LP code pairs supporting cup-product CCZ gates, including bivariate tricycle–bicycle families such as the $\llbracket 81,3,5\rrbracket$–$\llbracket 54,2,6\rrbracket$ pair, where the 3D tricycle codes admit depth-2 CCZ, weight-6 stabilizers, and pseudo-thresholds $\gtrsim 0.4\%$. 
As a byproduct, we show that the 3D codes enable highly efficient magic-state preparation: a single round of stabilizer measurements followed by depth-2 CCZ and postselection produces states with error $<10^{-9}$ and success probability $\sim 35\%$.
Our results establish a native integration of qLDPC codes with complementary transversal gates—covering nearly all practically relevant families known so far—and open a broad design space for scalable, low-overhead universal quantum computation.
\end{abstract}

\maketitle

\section{Introduction}
Quantum low-density parity-check (qLDPC) codes are strong candidates for scalable fault-tolerant quantum computation (FTQC) thanks to their low encoding overhead~\cite{gottesman2013fault, breuckmann2021quantum, panteleev2022asymptotically, leverrier2022quantum, gu2022efficient, dinur2022good, lin2022good, bravyi_high-threshold_2024, xu2024constant}.
Although low-overhead Clifford gates have been constructed for several practically relevant qLDPC families~\cite{cohen2022low, cross2024improved, xu2025fast, malcolm2025computing}, efficient implementations of non-Clifford gates in these regimes remain elusive. Leading approaches still rely on interfacing with low-rate topological codes to prepare high-fidelity magic states~\cite{yoder2025tour, zhang2025constant}, incurring large space-time overhead for computations with heavy non-Clifford demands. To unlock the genuine resource savings promised by qLDPC codes, it is therefore essential to develop native implementations of both Clifford and non-Clifford gates---ideally transversal ones~\cite{zhu2023non, lin2024transversal, golowich2025asymptotically, zhu2025topological, he2025quantum}.

Recent no-go results~\cite{fu2025no}, analogous to the Bravyi–König theorem for topological codes~\cite{bravyi2013classification}, 
show that certain qLDPC codes require a chain-complex dimension of at least three to support transversal (or constant-depth) non-Clifford gates. 
At the same time, in compliance with the Eastin–Knill theorem~\cite{eastin2009restrictions}, 2D codes remain necessary for complementary transversal Clifford gates. 
These insights motivate new schemes that enable dimension jumps between chain-complex dimensions --- analogous to dimension-jump protocols for topological codes in different spatial dimensions~\cite{bombin2016dimensional, heussen2024efficient, daguerre2025code} --- thereby enabling the combination of qLDPC codes with complementary native gates for universal, high-rate quantum fault tolerance.

In this work, we present a \emph{transversal dimension jump} scheme between different chain-complex dimensions that applies to a broad family of qLDPC codes --- the lifted product (LP) codes~\cite{eczoo_lifted_product, panteleev2021degenerate, panteleev2021quantum, breuckmann2021quantum}, which encompasses hypergraph product codes~\cite{tillich2014quantum}, quasi-cyclic lifted product codes~\cite{panteleev2021degenerate, panteleev2021quantum}, bivariate-bicycle codes~\cite{bravyi_high-threshold_2024}, abelian two-block group-algebra codes~\cite{kalachev2020minimum, wang2023abelian, lin2024quantum, eczoo_2bga}, etc. 
Our scheme is structural: given any 3D LP code $\QABC$ built from the lifted product of three classical codes $\mc{C}_A, \mc{C}_B, \mc{C}_C$, we can implement a one-way transversal logical CNOT circuit from $\QABC$ to any of its 2D component codes obtained by the lifted product of two of the three base classical codes (e.g. $\QAB$ from $\mc{C}_A$ and $\mc{C}_B$). Note that the code dimension here refers to the dimensionality of the resulting chain complex (see Appendix~\ref{sec:appx_background}), rather than to geometric locality.

Conceptually, as illustrated in Fig.~\ref{fig:overview}(a), the scheme works since the 3D code, arranged on a cube, can be viewed as stacks of multiple 2D slices, each structurally resembling a 2D component code. 
Utilizing the logical CNOT circuit, we can then perform teleportation-based code switching between the 3D code and multiple copies of the 2D component code \emph{in parallel} (see Fig.~\ref{fig:overview}(b-c)), when allowing geometrically nonlocal gates. 

This generic dimension-jump scheme yields powerful consequences: by pairing 3D LP codes supporting transversal/constant-depth CCZ gates with 2D LP codes supporting low-overhead Clifford gates, both with high encoding rates, and linking them via the dimension jump, one obtains universal fault-tolerant gatesets with low space-time costs.  
For example, 3D hypergraph product codes ---  the simplest instances of the LP codes --- with asymptotically constant rate and low-weight stabilizers~\cite{golowich2025quantum} can be switched to 2D HGP codes supporting addressable transversal Clifford gates~\cite{xu2025fast, berthusen2025automorphism, malcolm2025computing}, providing a high-rate transversal universal gate set. 

As case studies, we identify new instances of practically relevant codes, such as a family of finite-size bivariate tricycle codes---for example, the $\llbracket 81,3,5\rrbracket $ code---analogous to, yet outperforming, the trivariate tricycle codes in Ref.~\cite{jacob2025single}. These codes support depth-2 CCZ gates via cup-product constructions~\cite{breuckmann_cups_2024} and can switch to bivariate bicycle codes~\cite{bravyi_high-threshold_2024} through the transversal dimension jump. 
Remarkably, they feature low-weight ($6$) stabilizers, pseudo-thresholds about $0.4\%$, and large qubit savings ($\sim 5\times$) over standard 3D toric codes, when assuming no geometric locality constraints on quantum gates. 
As a byproduct, we show that, utilizing their 3D structure, these codes further enable a simple yet highly efficient magic-state preparation scheme: initialize logical $\ket{+}$ states, measure stabilizers once, apply the depth-2 CCZ circuit, and postselect using error detection. For the $\llbracket 81,3,5\rrbracket $ code, we estimate that this procedure yields magic states with (postselected) logical error rate below $10^{-9}$, success probability around $35\%$, and space–time overhead significantly lower than the best known cultivation protocols with color codes~\cite{gidney2024magic}.
Finally, we show that the higher-rate 3D codes of Ref.~\cite{menon2025magic} (e.g., the $\llbracket 108,6,6\rrbracket$ code), which support a modified cup-product CCZ circuit, can be readily paired and code-switched with their 2D component codes within our framework, highlighting the broad applicability of our scheme.

\section{Transversal dimension jump}
\begin{figure*}
    \centering
    \includegraphics[width=1\linewidth]{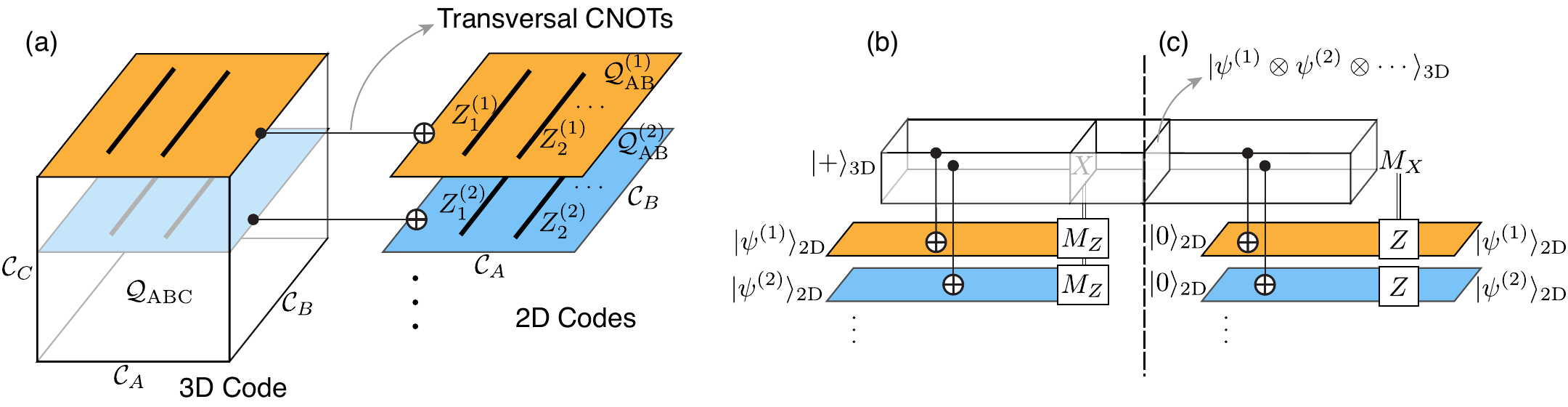}
    \caption{
    \textbf{Transversal dimension jump.}
    (a) A 3D LP code $\QABC$, defined as the lifted product of base classical codes $\mc{C}_A$, $\mc{C}_B$, and $\mc{C}_C$, can be visualized on a cube. Any of its 2D component codes, such as $\QAB$ (the product of two base codes), corresponds to a 2D slice of $\QABC$. 
    A logical CNOT circuit can be realized via transversal physical CNOTs between $\QAB$ and the slice of $\QABC$. 
    To see the logical action, the logical $Z$ operators of the 2D codes, indicated by the string-like operators, are mapped through the CNOT circuit to the corresponding $Z$ operators of the 3D code on the associated slice. 
    (b-c) Consequently, a subset of logical qubits of $\QABC$ can be teleported from (panel (b)) or to (panel (c)) one or more copies of the 2D codes in parallel, provided the CNOT circuits are also logically transversal (with one-to-one couplings between the logical qubits). The logical teleportation circuits are one-bit teleportation circuits that involve the transversal CNOTs, transversal logical initialization in the $X$ (resp. $Z$) basis, transversal measurements in the $Z$ (resp. $X$) basis --- denoted by $M_Z$ (resp. $M_X$) --- and corresponding Pauli $X$ (resp. $Z$) feedback corrections.
    }
    \label{fig:overview}
\end{figure*}
Here, we introduce in detail the \emph{transversal dimension-jump scheme} for a 3D LP code and its 2D component codes.

Let $G$ be a finite abelian group of order $l$, and let $R = \mathbb{F}_2[G]$ denote its group algebra. A classical code $\mc{C}$ over $R$ can be expressed as a $1$-chain complex
\begin{equation}
\mc{C}: 
    \begin{tikzcd} 
    C_1 = R^n & C_0 = R^r
\arrow["{H \in R^{r\times n}}", from=1-1, to=1-2],
\end{tikzcd}
\end{equation}
where $C_1$ and $C_0$ are associated the checks and bits of $\mc{C}$, respectively, and $H$ is the check matrix.  

Since $R \simeq \mathbb{F}_2^l$ as an $\mathbb{F}_2$-vector space with basis $\{g\}_{g \in G}$, the code admits an equivalent binary representation
$
\mc{C}:\quad 
    \begin{tikzcd}[column sep=4.5em]
    C_1 = \mathbb{F}_2^{nl} \arrow[r, "\mathbb{B}(H) \in \mathbb{F}_2^{rl\times nl}"] & C_0 = \mathbb{F}_2^{rl},
\end{tikzcd}
$
where the binary representation of a matrix $M$ over $R$, $\mbb{B}(M)$, is given by replacing each entry $M_{i, j} \in R$ with its faithful representation $\mbb{B}(M_{i,j}) \in \mbb{F}_2^{l\times l}$.
This $1$-complex $\mc{C}$ can be regarded as a 1D LP code.

The tensor product of two such classical codes $\mc{C}_A$ and $\mc{C}_B$ over $R$ yields a 2D LP code $\QAB = \mathrm{LP}(\mc{C}_A,\mc{C}_B)$ with a $2$-chain complex
\begin{equation}
\begin{tikzcd} 
    \QAB: Q_2^{\mr{AB}} & Q_1^{\mr{AB}} & Q_0^{\mr{AB}}
\arrow["\partial_2^{\mr{AB}}", from=1-1, to=1-2]
\arrow["\partial_1^{\mr{AB}}", from=1-2, to=1-3],
\end{tikzcd}
\end{equation}
where $Q_\ell^{\mr{AB}} = \bigoplus_{i+j=\ell} C_i^A \otimes_R C_j^B$ for $\ell=0,1,2,$ are associated with $X$-checks, qubits, and $Z$-checks, respectively. The tensor product is over $R$. 
The boundary maps act by $
\partial_\ell^{\mr{AB}}(a\otimes_R b) = (H_A a)\otimes_R b + a\otimes_R (H_B b)$ for $a \in C_i^A$ and $b \in C_j^B$ for $\ell=1,2$, forming the $X$ check matrix and the transpose of the $Z$-check matrix, respectively. Note that, we can again obtain the binary form of the check matrices using their binary representations.

Iterating once more, we obtain a 3D LP code $\QABC = \mathrm{LP}(\QAB,\mc{C}_C)$ with a $3$-chain complex:
\begin{equation}
    \begin{tikzcd} 
    \QABC: Q_3^{\mr{ABC}} & Q_2^{\mr{ABC}} & Q_1^{\mr{ABC}} & Q_0^{\mr{ABC}}
\arrow["{\partial_3^{\mr{ABC}}}", from=1-1, to=1-2]
\arrow["{\partial_2^{\mr{ABC}}}", from=1-2, to=1-3],
\arrow["{\partial_1^{\mr{ABC}}}", from=1-3, to=1-4]
\end{tikzcd}
\label{eq:3D_chain_complex}
\end{equation}
where $Q_\ell^{\mr{ABC}} = \bigoplus_{i+j=\ell} Q_i^{\mr{AB}}\otimes_R C_j^C$ for $\ell=0,1,2,3,$ are associated with $X$ checks, qubits, $Z$ checks, and meta-$Z$ checks (redundancy of the $Z$ checks), respectively. 
The boundary maps are given by $
\partial_\ell^{\mr{ABC}}(q\otimes_R c) = \partial_i^{\mr{AB}}(q)\otimes_R c + q\otimes_R H_C(c)$ for $q \in Q_i^{\mr{AB}}$ and $c \in C_j^C$ for $\ell=1,2,3,$ forming the $X$ check matrix, the transpose of the $Z$ check matrix, and the transpose of the meta-$Z$ check matrix, respectively.
See Appendix~\ref{sec:appx_codedetails} for details of these codes.

By symmetry, $\QABC$ can also be realized as $\mathrm{LP}(\QAC,\mc{C}_B)$ or $\mathrm{LP}(\QBC,\mc{C}_A)$, or equivalently as the triple lifted product $\mathrm{LP}(\mc{C}_A,\mc{C}_B,\mc{C}_C)$~\cite{breuckmann_quantum_2021}.
For this reason, we refer to $\{\QAB,\QAC,\QBC\}$ as the \emph{2D component codes} of $\QABC$.

Each 2D component code naturally embeds into $\QABC$. For instance,
\begin{equation}
    \QAB \simeq \tilde{Q}_{\mr{AB}} = \mathrm{LP}(\QAB,\mc{C}_0) \subset \QABC,
\end{equation}
where $\mc{C}_0: 0\to R^1$ is a trivial subcomplex of $\mc{C}_C$. Thus both the qubits and the checks of $\tilde{Q}_{\mr{AB}}$ can be viewed as subsets of those of $\QABC$ (see Fig.~\ref{fig:overview}(a)). This inclusion is expressed by a chain map
\begin{equation}
    \gamma:\ \tilde{Q}_{\mr{AB}}\hookrightarrow \QABC,
\end{equation}
consisting of $R$-linear maps $\gamma_i:\tilde{Q}_i^{\mr{AB}}\to Q_i^{\mr{ABC}}$ commuting with the boundary maps of $\tilde{Q}_{\mr{AB}}$ and $\QABC$. 
Different choices of $\gamma$ correspond to embedding $\tilde{Q}_{\mr{AB}}$ into different subsets of $\QABC$.

We now show that a physical CNOT circuit specified by $\gamma_1$ induces a \emph{logical} CNOT circuit between $\tilde{Q}_{\mr{AB}}$ and $\QABC$.
Specifically, we apply a physical CNOT controlled by the $i$-th qubit of $\QABC$ and targeting the $j$-th qubit of $\tilde{Q}_{\mr{AB}}$ whenever $\mathbb{B}(\gamma_1)_{i,j}=1$. 
By Refs.~\cite{huang2022homomorphic,xu2025fast}, such a circuit implements a valid logical CNOT circuit if and only if $\gamma$ is a chain map (also called homomorphism) ---guaranteed here by construction--- and is termed homomorphic CNOT:

\begin{lemma}[Homomorphic CNOT]
    For any 3D LP code $\QABC=\mathrm{LP}(\mc{C}_A,\mc{C}_B,\mc{C}_C)$, there exists a logical CNOT circuit controlled by $\QABC$ and targeting any of its 2D component codes, implemented entirely via transversal physical CNOTs.
    \label{lemma:homo_CNOT}
\end{lemma}
See Appendix~\ref{sec:appx_proofhomomorphiccnot} for a detailed proof. 
The homomorphic CNOT is directional: it can be controlled by the 3D code but not the other way around. Moreover, it is physically transversal, as each physical qubit interacts with at most one partner qubit due to the inclusion-map structure of $\gamma_1$. We therefore also refer to it as a one-way transversal CNOT, a construction that is inherently fault-tolerant.

At the logical level, $\gamma_1$ induces a map on first homologies
\begin{equation}
    \bar{\gamma}_1:H_1(\tilde{Q}_{\mr{AB}})\longrightarrow H_1(\QABC),
\end{equation}
which specifies how logical $Z$ operators of $\tilde{Q}_{\mr{AB}}$ are mapped to those of $\QABC$ through the logical CNOT circuit (see, e.g. the string-like logical $Z$ operators in Fig.~\ref{fig:overview}(a)). Concretely, given a basis 
of $k_{\mr{AB}}$ and $k_{\mr{ABC}}$ logical $Z$ operators for $H_1(\tilde{\mc{Q}}_{\mr{AB}})$ and $H_1(\QABC)$ respectively, the induced map is represented by a matrix $\bar{\gamma}_1\in \mathbb{F}_2^{k_{\mr{ABC}}\times k_{\mr{AB}}}$. 
A logical CNOT controlled by the $i$-th logical qubit of $\QABC$ and targeting the $j$-th of $\tilde{Q}_{\mr{AB}}$ is thus applied whenever $\bar{\gamma}_{1_{i,j}}=1$.

If $\bar{\gamma}_1$ is injective, which implies that $\bar{\gamma}_1$ can be written in a transversal form (in certain logical basis)
\begin{equation}
    \bar{\gamma}_1=\begin{pmatrix} I_{k_{\mr{AB}}} \\ 0_{(k_{\mr{ABC}} - k_{\mr{AB}})\times k_{\mr{AB}}}\end{pmatrix},\qquad (k_{\mr{AB}}\leq k_{\mr{ABC}}),
\end{equation}
then every logical qubit of the 2D code couples distinctly to one of the 3D code --- the homomorphic CNOT circuit is also \emph{logically transversal}. 
This provides a sufficient condition for a transversal dimension-jump scheme: 
by combining such a transversal logical CNOT circuit with the logical teleportation circuit of Fig.~\ref{fig:overview}(b–c), we can teleport logical qubits between $\QABC$ and $\tilde{Q}_{\mr{AB}}$.

\begin{theorem}[Transversal dimension jump]
    Given any 3D LP code $\QABC = \mr{LP}(\mc{C}_A, \mc{C}_B, \mc{C}_C)$, one can teleport a subset of its logical qubits to or from a 2D component code in parallel using the homomorphic CNOT of Lemma~\ref{lemma:homo_CNOT}, provided they are logically transversal.
    \label{th:transversal_dim_jump}
\end{theorem}

The sufficient condition for logical transversality---the injectivity of $\bar{\gamma}_1$---is mild. As shown in the next section, it holds for all HGP codes (with $R = \mbb{F}_2$) and for every nontrivial 3D LP code we identify in Table~\ref{table:codes} that supports low-depth non-Clifford gates.

Moreover, this transversal teleportation can be extended to many copies of a 2D component code $\mc{Q}_{\mr{2D}}^{(i)}$ simultaneously. 
This requires constructing physical CNOTs $\gamma_1^{(i)}: Q_1^{\mr{2D},(i)} \rightarrow Q_1^{\mr{ABC}}$ with disjoint images, 
so that each 2D code couples to a distinct subset of the 3D code, as illustrated in Fig.~\ref{fig:overview}(a).

\section{Case studies \label{sec:case_studies}}
We now present case studies of several subfamilies of LP codes, demonstrating that the transversal dimension jump (Theorem~\ref{th:transversal_dim_jump}) can be applied to a broad class of practically relevant high-rate 2D and 3D code pairs, all encoding more than one logical qubit. 

Our focus will be on 3D LP codes that admit low-depth logical CCZ circuits via the cup-product construction~\cite{breuckmann_cups_2024}. 
At a high level, given three copies of a 3D LP code $\{\mc{Q}_{\mr{3D}}^{(\alpha)}\}_{\alpha=1,2,3}$, a cross-block physical CCZ circuit is specified by a 3D tensor $\delta$. 
This circuit applies a physical CCZ gate to qubits $i$, $j$, and $k$ of $\mc{Q}_{\mr{3D}}^{(1)}$, $\mc{Q}_{\mr{3D}}^{(2)}$, and $\mc{Q}_{\mr{3D}}^{(3)}$, respectively, whenever $\delta_{i,j,k}=1$. 
The circuit $\delta$ defines a valid logical CCZ circuit if it induces a trilinear map on cohomologies, 
$
\bar{\delta}: H^1(\mc{Q}_{3D}^{(1)}) \times H^1(\mc{Q}_{3D}^{(2)}) \times H^1(\mc{Q}_{3D}^{(3)}) \rightarrow \mbb{F}_2
$.
Given a basis of the cohomologies (corresponding to $X$ logical operators), the tensor $\bar{\delta}$ specifies a logical CCZ circuit in the same way that $\delta$ specifies the physical one. 
The logical CCZ is nontrivial whenever $\bar{\delta}$ is not the identity tensor.

\begin{table}
\centering
% \begin{tabular}{l|c|c|c|c|c}
\begin{tabular}{c|c|c|c|c}
\toprule
\textbf{Code} & \textbf{2D} $\llbracket n,k,d\rrbracket$ & \textbf{3D} $\llbracket n,k,d\rrbracket $ & \textbf{3D} $kd^3/n$ & \textbf{CCZ depth} \\
\midrule
Pentagon & $\llbracket 45, 7, 3\rrbracket $ & $\llbracket 180, 8, 3\rrbracket $ & 1.2 & 4 \\
Lifted toric & $\llbracket 16, 2, 4\rrbracket $ & $\llbracket 48, 3, 4\rrbracket $ & 4 & 2 \\
Lifted toric & $\llbracket 36, 2, 6\rrbracket $ & $\llbracket 162, 3, 6\rrbracket $ & 4 & 2 \\
BB/BT & $\llbracket 18, 2, 3\rrbracket $ & $\llbracket 27, 3, 3\rrbracket $  & 3 & 2 \\
BB/BT  & $\llbracket 30, 2, 5\rrbracket $ & $\llbracket 45, 3, 4\rrbracket $ & 4.3 & 2 \\
BB/BT  & $\llbracket 54, 2, 6\rrbracket $ & $\llbracket 81, 3, 5\rrbracket $ & 4.6 & 2 \\
TB/TT  & $\llbracket 140, 2, 8\rrbracket $ & $\llbracket 210, 3, 7\rrbracket $ & 4.9 & 2 \\
TB/TT \cite{menon2025magic} & $\llbracket 72,4,6\rrbracket$ & $\llbracket 108, 6, 6\rrbracket$ & 12 & 8 \\
TB/TT \cite{menon2025magic} & $\llbracket 160,4,8\rrbracket$ & $\llbracket 240, 6, 8\rrbracket$ & 12.8 & 8 \\
\bottomrule
\end{tabular}
\caption{\textbf{List of 2D–3D code pairs switchable via transversal dimension jump}. For each 3D code, we report its qubit savings $kd^3/n$ relative to the standard 3D toric code and the physical CCZ depth for implementing the logical CCZ circuit via the cup-product construction~\cite{breuckmann_cups_2024}. 
Let $(r,c)$ denote the row- and column-weight of a matrix. All (but the Pentagon) 2D codes have $(4, 2)$ ($(6, 4)$) check matrices; all (but the Pentagon) 3D codes have $(4, 2)$ ($(6, 4))$ $X$ check matrices and $(6, 4)$ ($(8, 6)$) $Z$ check matrices. See Appendix~\ref{sec:appx_codedetails} for the explicit constructions of these codes. 
The 3D TT codes in the last two code pairs are constructed in Ref.~\cite{menon2025magic} using a modified, more symmetric form of the cup-product CCZ circuits, which enables higher encoding rates.
}
\label{table:codes}
\end{table}

\subsection{Hypergraph product codes}
Taking $G$ to be the trivial group with only the identity element gives $R \simeq \mbb{F}_2$ and recovers the hypergraph product (HGP) codes~\cite{tillich_quantum_2013}. 

For any 3D HGP code $\QABC = \mr{HGP}(\mc{C}_A, \mc{C}_B, \mc{C}_C)$ with three base codes over $\mbb{F}_2$, we show that transversal teleportation can be performed to (multiple copies of) any of its 2D component codes.  

Without loss of generality, consider the component code $\QAB = \mr{HGP}(\mc{C}_A, \mc{C}_B)$ and write $\QABC = \mr{HGP}(\QAB, \mc{C}_C)$. By the Künneth formula~\cite{hatcher2005algebraic}, the first homology of $\QABC$ can be derived from those of $\QAB$ and $\mc{C}_C$:
\begin{equation}
    H_1(\QABC) = H_1(\QAB)\otimes H_0(\mc{C}_C) \oplus H_0(\QAB)\otimes H_1(\mc{C}_C),
    \label{eq:HGP_kunneth}
\end{equation}
where $H_1(\mc{C}_C)$ (resp. $H_0(\mc{C}_C)$) denotes the kernel (resp. cokernel) of $H_C \in \mbb{F}_2^{r_C\times n_C}$ --- the check matrix of $\mc{C}_C$.  

Up to column permutations, we can decompose $H_C^T$ as 
$H_C^T = \big(H_C^{T,0},\, H_C^{T,1}\big)$,
where $H_C^{T,0} \in \mbb{F}_2^{n_C \times k_C^T}$, and $H_C^{T,1} \in \mbb{F}_2^{n_C \times (r_C - k_C^T)}$ have linearly independent columns, where $k_C^T$ denotes the dimension of the transpose code $H_C^T$. Now, we have
\begin{equation}
    H_0(\mc{C}_C) = \mr{rs}^{\bullet}(H_C^T) = \mr{span}\{e^{r_C}_i\}_{i \in [k_C^T]},
\end{equation}
where $\mr{rs}^{\bullet}$ denotes the complementary space of the row space and $e^{r_C}_i$ denotes a unit column vector in $\mbb{F}_2^{r_C}$ with the $i$-th entry being $1$.
Substituting into Eq.~\eqref{eq:HGP_kunneth}, the first sector of $H_1(\QABC)$ simplifies to  
\begin{equation}
    H_1(\QAB)\otimes H_0(\mc{C}_C) 
    = \mr{span}\{ H_1(\QAB)\otimes e^{r_C}_i \}_{i \in [k_C^T]},
\end{equation}
which shows that the first-sector logical $Z$ operators of $\QABC$ are supported on $k_C^T$ disjoint slices of its physical qubits, each containing an identical copy of the $k_{\rm AB}$ logical operators of $\QAB$ (see Fig.~\ref{fig:overview}(a)).  

Now, for each copy $\QAB^{(i)}$ with $i \in [k_C^T]$, choosing 
\begin{equation}
    \gamma_1^{(i)} = \left( 
    \begin{array}{c}
         I_{n_{\rm AB}} \otimes e^{r_C}_i  \\
         0
    \end{array} \right):
    Q_1^{\mr{AB, (i)}} \rightarrow Q_1^{\mr{ABC}} = Q_1^{\mr{AB}}\otimes C_0^C \oplus \cdots,
\end{equation}
which transversally couples the $n_{\mr{AB}} = \mr{dim}(Q^{\mr{AB}}_1)$ physical qubits of $\QAB^{(i)}$ to the $i$-th slice of the first-sector qubits of $\QABC$, yields an injective logical map $\bar{\gamma}_1^{(i)}: H_1(\QAB^{(i)}) \to H_1(\QABC)$, corresponding to logically transversal CNOTs.
Moreover, since the images of the $\{\gamma_1^{(i)}\}$ are disjoint, corresponding to coupling different 2D codes to different slices of the 3D code, transversal teleportation can be performed between $\QABC$ and up to $k_C^T$ 2D component codes in parallel.  

As a concrete example, let $H_A$ and $H_C$ be the $[3,1,3]$ repetition code, represented on a cycle graph with checks at vertices and bits on edges. Let $H_B$ be a $[10,6,3]$ code embedded on a fully connected graph with five vertices and ten edges, which can be visualized as a pentagon.
These choices yield a $\llbracket 180, 8, 3 \rrbracket$ 3D code $\QABC$ and a $\llbracket 45, 7, 3 \rrbracket$ 2D code $\QAB$ (the Pentagon codes in Table~\ref{table:codes}). 
Since $H_C^T$ is again the $[3,1,3]$ repetition code, we have $k_C^T = 1$, and thus only a single copy of $\QAB$ can be coupled to $\QABC$ in this case.  

Finally, because the classical base codes are defined on graphs, which form simplicial complexes, they inherit a natural cup product. This cup product can be lifted to the quantum code via tensor products~\cite{breuckmann_cups_2024}, enabling the implementation of logical CCZ gates by a physical CCZ circuit of depth $4$.

\subsection{Univariate lifted product codes}
Taking \(G = C_l\), the cyclic group of order \(l\), we have \(R \simeq \mathbb{F}_2[x]/(x^l + 1)\), yielding LP codes over (quotient) univariate polynomials. This family is conventionally referred to as quasi-cyclic LP codes~\cite{panteleev2021degenerate, panteleev2021quantum}.

One systematic way to obtain such univariate LP codes is by ``lifting": starting from a parental HGP code with binary classical check matrices and replacing their binary entries with univariate polynomials. 
A natural strategy to construct univariate LP codes with cup-product CCZ gates is to lift a parental HGP code that already supports such a gate.

As concrete examples, Table~\ref{table:codes} lists two 3D univariate LP codes with parameters \( \llbracket 48, 3, 4\rrbracket \) and \( \llbracket 162, 3, 6\rrbracket \), obtained by lifting the \( \llbracket 24, 3, 2\rrbracket \) and \( \llbracket 81, 3, 3\rrbracket \) 3D toric codes with base repetition codes over \(\mathbb{F}_2[x]/(x^2 + 1)\), respectively. 
For instance, choosing
$H_A = H_B = H_C =
\begin{pmatrix}
1 & x \\[2pt]
1 & 1
\end{pmatrix}
$
yields the \( \llbracket 48, 3, 4\rrbracket \) code.

Here, the lifting operation doubles the code distance while enlarging the code size only by a factor of $2$---in contrast to a factor of $8$ for standard 3D toric codes---yielding a $4\times$ saving in physical qubits. As shown in Table~\ref{table:codes}, these lifted toric codes can couple to 2D component codes with parameters also superior to those of corresponding 2D surface or toric codes via logically transversal homomorphic CNOTs, as verified numerically.

Finally, these 3D codes support a depth‑2 logical CCZ circuit via a depth‑2 physical CCZ circuit---matching that for standard 3D toric codes.

\subsection{Multivariate bicycle/tricycle codes}
Taking \(G\) to be a direct product of cyclic groups, we have \(R = \mathbb{F}_2[G]\) isomorphic to multivariate polynomial rings.  
For example, \(G = C_{l_x} \times C_{l_y}\) (resp. \(C_{l_x} \times C_{l_y} \times C_{l_z}\)) gives
$R_{l_x,l_y} \simeq \mathbb{F}_2[x,y]/(x^{l_x}+1, y^{l_y}+1)$ (resp. $R_{l_x,l_y,l_z} \simeq \mathbb{F}_2[x,y,z]/(x^{l_x}+1,y^{l_y}+1,z^{l_z}+1)
$).

Choosing base check matrices \(H_A=a, H_B=b, H_C=c\), where $a, b, c \in R_{l_x, l_y}$ (resp. $R_{l_x, l_y, l_z}$) yields bivariate (resp. trivariate) bicycle (2D) and tricycle (3D) codes, which we refer to as BB/BT/TB/TT codes for short.
Note that this family covers many BB~\cite{bravyi_high-threshold_2024} and TT~\cite{jacob2025single} codes that have high encoding rates and promising performance.

Our homomorphic CNOTs in Lemma~\ref{lemma:homo_CNOT} then couples bicycle-tricycle code pairs over the same polynomial ring $R$. 
As detailed in Appendix~\ref{sec:appx_inclusionmaplogicaltransversality}, we show that taking the homomorphic CNOTs
\begin{equation}
    \gamma_1 = \left( \begin{array}{c}
     I_{n_{\mr{AB}}} \\
     0
\end{array}\right): Q_1^{\mr{AB}} \rightarrow Q_1^{\mr{ABC}} \simeq Q_1^{\rm{AB}}\oplus Q_0^{\mr{AB}},
\end{equation}
which transversally couples the physical qubits of the 2D code with those of the first sector of the 3D code,
the induced logical CNOTs $\bar{\gamma}_1: H_1(\QAB) \rightarrow H_1(\QABC)$ are then injective if:
(1) the group order $|G|$ is odd, and (2) $c \in (a, b)$ over $R$.
Here, $n_{\mr{AB}} = \mr{dim}(Q^{\mr{AB}}_1) = 2$ (as $R$-modules).

As an explicit case, choosing 
\begin{equation}
    c = x+x^2y,\ a = x^2y+x^2y^2,\ b=1+xy^2,
\end{equation}
over a bivariate polynomial ring $R_{3, 3}$ gives the \(\llbracket 27,3,3\rrbracket\) code in Table~\ref{table:codes}.
Since \(c=x^2a+x^2b\in(a,b)\), the \(\llbracket 18,2,3\rrbracket\) 2D component code \(\QAB\) couples transversally to two of the three logical qubits of \(\QABC\).

All multivariate tricycle codes we constructed with binomial polynomials \(a,b,c\) support depth‑2 cup‑product CCZ gates~\cite{breuckmann_cups_2024, jacob2025single}, analogous to standard 3D toric codes.
We have numerically confirmed that the logical CCZ circuits for all the tricycle codes in Table~\ref{table:codes} are nontrivial. We also included higher distance examples of trivariate tricycle codes from \cite{menon2025magic} in the last two rows of Table~\ref{table:codes}. These codes are constructed using the symmetric triple cup product, analogous to the method from \cite{breuckmann_cups_2024}, to also have constant-depth logical CCZ while having higher encoding rates.
For each of these 3D TT codes, we can pair it with a 2D component TB code (listed in Table~\ref{table:codes}), between which the logical CNOTs are transversal, and our transversal dimension-jump protocol is therefore applicable, as verified numerically.
See Appendix~\ref{sec:appx_cupproduct} for an outline of the construction of these codes and Ref.~\cite{menon2025magic} for details and more instances of 3D codes. 

\section{Numerical simulation \label{sec:numerics}}
\begin{figure}
    \centering
    \includegraphics[width = 0.5\textwidth]{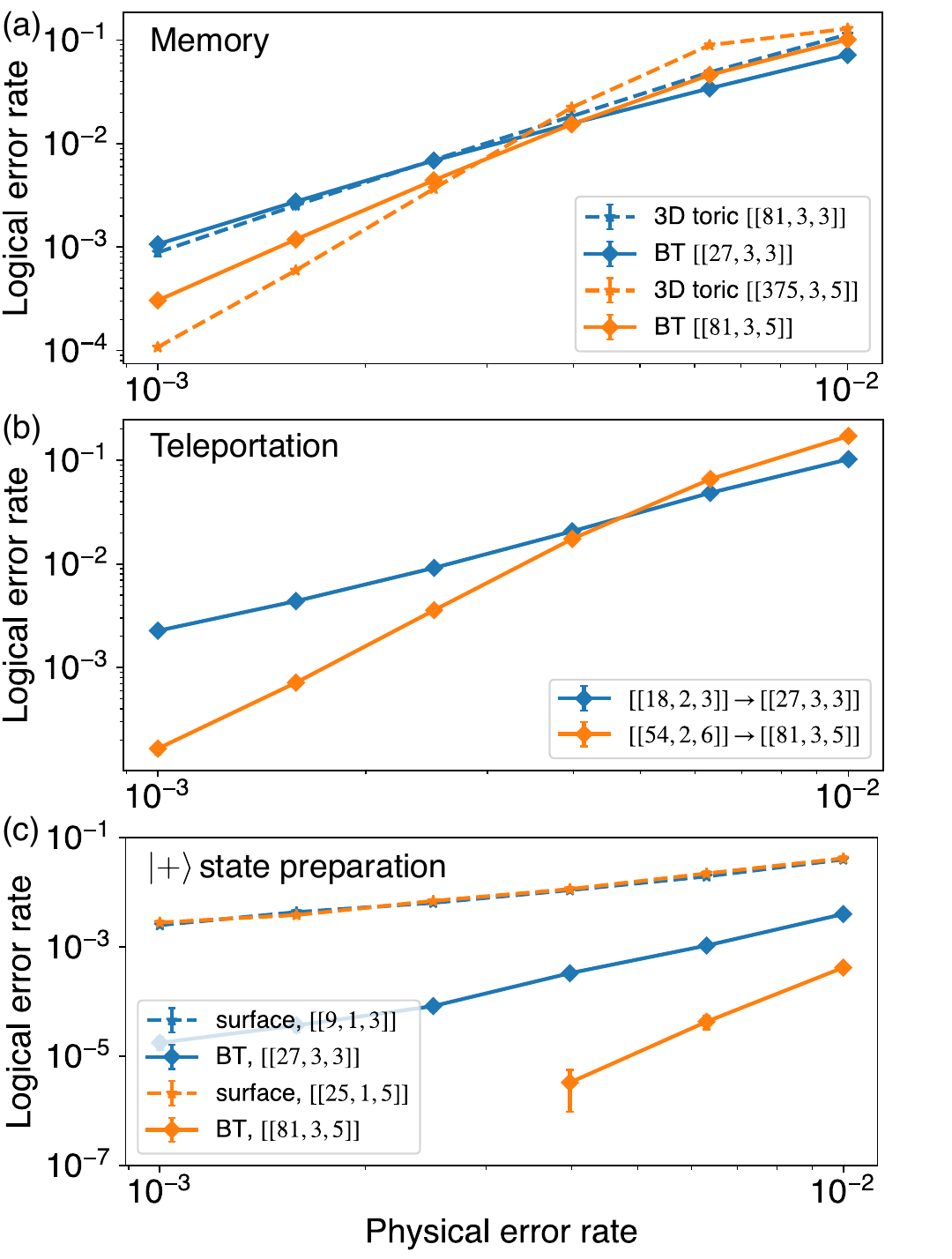}
    \caption{\textbf{Logical error rates of 3D codes and the transversal dimension jump circuit}. 
    (a) Memory logical error rate per logical qubit per code cycle for the BT codes in Table~\ref{table:codes}, compared against standard 3D toric codes. 
    (b) Logical error rate per logical qubit for the transversal teleportation circuit in Fig.~\ref{fig:overview}(b) using the BB/BT code pairs in Table~\ref{table:codes}. 
    (c) Logical error rate per logical qubit for single-shot $|+\rangle$ state preparation using the BT codes, compared against 2D surface codes. }
    \label{fig:simulations}
\end{figure}

\begin{table}
\centering
\begin{tabular}{c|c|c|c|c}
\toprule
Code & Parameters & Success rate & LER & Space-time cost \\
\midrule
BT & $\llbracket 27, 3, 3\rrbracket $ & $0.7$ & $1.(9) \times 10^{-6}$ & $239$ \\
BT  & $\llbracket 81, 3, 5\rrbracket $ & $0.35$ & $2.(3)\times 10^{-10}$ & $1444$ \\
Toric & $\llbracket 81, 3, 3\rrbracket $ & $0.35$ & $1.(9)\times 10^{-6}$ & $1444$ \\
Toric & $\llbracket 375, 3, 5\rrbracket $ & $0.008$ & $\backslash$ & $291667$ \\
\midrule
Color & $\llbracket 7,1,3\rrbracket$ & 0.65 & $6\times 10^{-7}$ & $540$ \\
Color & $\llbracket 19,1,5\rrbracket$ & 0.15 & $6\times 10^{-10}$ & $24000$ \\
\bottomrule
\end{tabular}
\caption{
\textbf{Magic state preparation with 3D codes and error detection.}
We report the success rate, postselected logical error rate (LER), and space‑time cost (qubit $\times$ expected circuit depth) per logical qubit for preparing magic states via a depth‑2 CCZ circuit applied to three copies of single‑shot prepared logical $|+\rangle$ states of BT codes, assuming a $10^{-3}$ physical error rate. Results are compared to the same protocol using standard 3D toric codes and the magic‑state cultivation approach for 2D color codes~\cite{gidney2024magic}. See Appendix~\ref{sec:appx_calcspacetimecost} for details on calculation of space-time cost.
}
\label{table:magic_states}
\end{table}

Here, we present numerical simulations of the bicycle/tricycle codes listed in Table~\ref{table:codes} under a standard circuit-level depolarizing noise model excluding idling errors.
In addition, we assume no geometric locality constraints on quantum gates. The simulation software used in this work are available at \href{https://github.com/tinghui2012/transversal-dimension-jump/}{Github}.

First, as shown in Fig.~\ref{fig:simulations}(a), 
the 3D tricycle codes exhibit high memory (pseudo-)thresholds of about \(0.4\%\), matching that for standard 3D toric codes.  
% Except for the pentagon code, 
They are decoded efficiently using a minimum‑weight perfect matching decoder for the dominant \(Z\)-type errors, since these codes have biased distances \(d_X \gg d_Z\) and $Z$ errors produce point‑like syndromes.

Second, as shown in Fig.~\ref{fig:simulations}(b), the transversal implementation of the homomorphic CNOTs---being inherently fault‑tolerant---ensures that the 2D–3D teleportation‑based code switching (circuit in Fig.~\ref{fig:overview}(b)) preserves both the high threshold and the low error rates of the underlying 3D codes.  
These results indicate that, since the CCZ circuit is shallow (depth-$2$) and introduces only limited error propagation (as shown in Table~\ref{table:magic_states}), the performance of universal computations—both in threshold and logical error rate—is well captured by the memory performance of the 3D codes, provided sufficiently high-fidelity Clifford gates on the 2D codes.

Finally, as a byproduct, we consider new magic‑state preparation protocols by applying a depth‑2 CCZ circuit to \(\overline{\ket{+}}^{\otimes 3k_{3\mr{D}}}\) on three copies of a 3D tricycle code \(\mathcal{Q}_{\mr{3D}}\), each encoding \(k_{\mr{3D}}\) logical qubits.
Note that similar protocols were presented and analyzed in Ref.~\cite{menon2025magic}.
Here, we focus on smaller codes with shallower CCZ circuits and error detection, which may be more relevant for near-term experiments.
As shown in Fig.~\ref{fig:simulations}(c), the logical $\ket{+}$ states can be prepared in single-shot by measuring only one round of $Z$ checks, with measurement errors corrected using the code’s meta-$Z$ checks enabled by its 3D structure (see Eq.~\ref{eq:3D_chain_complex}).
See also complementary works on single-shot aspects of the codes~\cite{jacob2025single, menon2025magic, tan2025single}.
To achieve high‑fidelity magic states with small codes, we perform error detection (postselection) on the circuit.  
As shown in Table~\ref{table:magic_states}, using small BT codes of distances \(3\) and \(5\) yields magic states with logical error rates \(\sim 10^{-6}\) and \(\sim 10^{-10}\), with high success rates of \(\sim 70\%\) and \(\sim 35\%\), respectively.  
These outperform standard 3D toric codes of the same distances due to their smaller block sizes; for example, a distance‑5 toric code would yield a negligible (\(\sim 1\%\)) success rate owing to its large block size.

Compared to the state‑of‑the‑art magic-state cultivation protocol using 2D color codes~\cite{gidney2024magic}, our approach achieves significantly lower space‑time cost per logical qubit for three reasons:  
1. Despite being 3D, our codes have comparable space overhead to the 2D topological codes in Ref.~\cite{gidney2024magic} thanks to their qLDPC structure;  
2. Magic states are prepared directly via transversal/low-depth non‑Clifford gates, avoiding repeated Clifford measurements with GHZ ancillas as in Ref.~\cite{gidney2024magic}, which require additional postselection and reduce success rates;  
3. Our protocol utilizes single‑shot state preparation, requiring only one code cycle, unlike multiple cycles (\(\sim d\)) in Ref.~\cite{gidney2024magic}.  

We note that these estimates are only approximate. Rather than simulating the full magic-state preparation with the physical CCZ gates, we only add approximate error channels resulting from the physical CCZ gates, e.g. correlated $Z$ errors propagated from $X$ errors before the CCZ circuit.
% and assume single-shot preparation of $\overline{\ket{+}}$ states in the 3D codes (as they have meta $Z$ checks), without fully analyzing time-like errors. See complementary work~\cite{jacob2025single, Tanetal} for analyz ing the single-shot properties of the 3D codes.
A full simulation is left for future work. We also note that the high-fidelity magic states considered here are stored in low-distance codes operating in error-detection mode; to enable fault-tolerant computation, they need to be transferred (``escaped''~\cite{gidney2024magic}) into large-distance codes, which we defer to future study.

\section{Conclusion and discussion}
In summary, we introduce a unified framework for transversal dimension jump in lifted product codes. 
By leveraging high-rate 3D and 2D LP codes with complementary sets of transversal logical gates and high thresholds, this paradigm offers a route to universal quantum computation that may bypass the need for costly magic-state distillation or cultivation protocols. 
% More broadly, it opens a wide design space for low-overhead quantum fault tolerance, where carefully chosen 3D-2D LP code pairs can be combined to realize scalable and practical architectures for transversal universal quantum computation. 
While preparing this manuscript, we became aware of complementary works~\cite{tan2025single, golowich2025constant} that focus on analogous code-switching protocols more specifically for HGP codes and on achieving universal, single-shot gate sets. Together with our present work, and with recent advances in 3D code constructions featuring native non-Clifford gates~\cite{jacob2025single, menon2025magic}, these developments collectively open a broad design space for low-overhead fault tolerance, where judiciously chosen 3D–2D LP code pairs can enable scalable and practical architectures for transversal universal quantum computation.
% Beyond our preliminary exploration, we expect that much more efficient codes can be constructed within this framework~\cite{menon2025magic}.

In Appendix~\ref{sec:appx_codedetails}, we show that certain 3D LP codes --- such as the bivariate tricycle codes --- can, despite their 3D structure, be embedded on a 2D torus with translationally invariant checks, similar to the bivariate bicycle codes~\cite{bravyi_high-threshold_2024}. This observation points to the possibility of implementing these codes efficiently on platforms such as multi-layer superconducting hardware~\cite{bravyi_high-threshold_2024, yoder2025tour} or reconfigurable atom arrays~\cite{bluvstein2023logical, xu2024constant, viszlai2024efficient}.

\begin{acknowledgments}

We thank Margarita Davydova, Pablo Bonilla, Varun Menon, Rohan Mehta, Roland Farrell, Samuel Tan, Yifan Hong and Guo Zheng for helpful discussions. We acknowledge financial support from the U.S. Department of Energy, Office of Science, National Quantum Information Science Research Centers, Quantum Systems Accelerator, and the National Science Foundation (PHY-2317110).
C. L. is funded in part by Caltech’s Summer Undergraduate Research Fellowship (SURF).
Q.X. is funded in part by the Walter Burke Institute for Theoretical Physics at Caltech.
The Institute for Quantum Information and Matter is an NSF Physics Frontiers Center.

\end{acknowledgments}

\newpage
\begin{appendix}
\appendix 

\section{Further background on lifted product codes \label{sec:appx_background}}
Let $G$ be an abelian group with order $l$. Its group algebra over $\mbb{F}_2$ forms a ring $R = \mbb{F}_2[G]$.

\begin{definition}[Chain complex over $R$ and its binary representation]
    Let $\{A_i\}_{i = 0}^D$ be a collection of $R$-modules and $\{\partial^A_i: A_i \rightarrow A_{i - 1}\}_{i = 1}^D$ be a collection of $R$-linear maps satisfying $\partial^A_{i - 1} \partial^A_i = 0$. They form the following $D$-chain complex $\mc{A}$ over $R$:
    \begin{equation}
    \mc{A}: 
        \begin{tikzcd}[cells={nodes={minimum height=2em}}]
    A_D \arrow[r,"\partial^A_D"] & A_{D-1} \arrow [r,"\partial^A_{D - 1}"]  & \cdots \arrow[r, "\partial^A_2"] & A_1 \arrow[r,"\partial^A_1"] & A_0
    \end{tikzcd},
    \end{equation}
    which we also denote as $\mc{A}=\{\{A_i\}_i,\{\partial^{A}_i\}_i\}$. Since $R \simeq \mbb{F}_2^l$ can be viewed as a $\mbb{F}_2$ vector space with basis $\{g\}_{g \in G}$, we can equivalently express $\mc{A}$ as a chain complex over $\mbb{F}_2$:
    \begin{equation}
    \mbb{B}(\mc{A}): 
        \begin{tikzcd}[cells={nodes={minimum height=2em}}]
    \mbb{B}(A_D) \arrow[r,"\mbb{B}(\partial^A_D)"] & \mbb{B}(A_{D-1}) \arrow [r,"\mbb{B}(\partial^A_{D - 1})"]  & \cdots \arrow[r,"\mbb{B}(\partial^A_1)"] & \mbb{B}(A_0)
    \end{tikzcd},
    \end{equation}
    where $\mbb{B}(A_i) = \mbb{F}_2^{m_i l}$ with $m_i$ being the dimension of $A_i$ as a $R$-module, and $\mbb{B}(\partial^A_i) \in \mbb{F}_2^{m_{i - 1}l \times m_i l}$ is the binary representation of $\partial^A_i \in R^{m_{i - 1}\times m_i}$ by replacing each entry $(\partial^A_i)_{p, q}$ with its faithful representation $\mbb{B}((\partial^A_i)_{p, q}) \in \mbb{F}_2^{l \times l}$.  
\end{definition}

\begin{definition}[Classical linear codes]
    A classical linear code over $R$ with a check matrix $H \in R^{r_C\times n_C}$ can be represented as a 1-chain complex $\mc{C}$:
\[\begin{tikzcd}[cells={nodes={minimum height=2em}}]
\mc{C}: C_1 = R^{r_C} \arrow[r,"H"] & C_0 = R^{n_C}
\end{tikzcd}\]
where the basis of $C_1$ and $C_0$ are the bits and checks, respectively. Its binary representation $\mbb{B}(\mc{C})$ gives the standard binary format of the code with a check matrix $H \in \mbb{F}_2^{r_C l \times n_C l}$.
\end{definition}

\begin{definition}[CSS quantum codes]
    A CSS quantum code over $R$ with $X$ and $Z$ check matrix $H_X \in R^{r_X \times n}$ and $H_Z \in R^{r_Z \times n}$, respectively, can be represented as a 2-chain complex $\mc{Q}$
\[\begin{tikzcd}[cells={nodes={minimum height=2em}}]
\mc{Q}: Q_2 = R^{r_Z} \arrow[r,"H_Z^T"] & Q_1 = R^n \arrow[r,"H_X"] & Q_0 = R^{r_X}
\end{tikzcd}\]
where the basis of $Q_2, Q_1$ $Q_0$ are the Z-type checks, qubits, and X-type checks, respectively. Its binary representation $\mbb{B}(\mc{Q})$ gives the standard binary format of the quantum code (for qubits) with check matrices $H_X \in \mbb{F}_2^{r_Xl \times nl}$ and $H_Z \in \mbb{F}_2^{r_Z l \times nl}$.
\end{definition}

\begin{definition}[Tensor product of chain complexes over $R$-modules \cite{hatcher2005algebraic}]
\label{def:product_complex}
Let $\mc{A} = \{\{A_i\},\{\partial^{A}_i\}\}$ and $\mc{B} = \{\{B_j\},\{\partial^{B}_j\}\}$ be $D_A$- and $D_B$-chain complexes over $R$, respectively. Their tensor product $\mc{A}\otimes_R \mc{B}$ is defined by a $(D_A + D_B)$-chain complex over $R$ with modules 
$\{(\mc{A}\otimes_R \mc{B})_n := \otimes_{i+j=n} A_i\otimes_R B_j\}_{n = 0}^{D_A + D_B}$ and boundary maps $\{\partial^{\mr{AB}}_n: (\mc{A}\otimes_R \mc{B})_n \rightarrow (\mc{A}\otimes_R \mc{B})_{n - 1}\}_{n = 1}^{D_A + D_B}$, where $\partial^{\mr{AB}}_n$ is defined as
\begin{equation}
    \partial^{\mr{AB}}_n (a_i \otimes_R b_j ) = \partial^{A}_i(a_i)\otimes b_j + a_i\otimes \partial^{B}_j(b_j),
\end{equation}
for $a_i\in A_i$ and $b_j\in B_j$ (with $i + j = n$).
This is the same as the total complex of a double (product) complex, otherwise known as a product total complex~\cite{breuckmann_quantum_2021}.
\end{definition}
\noindent \textbf{Note:} For simplicity, we drop the subscript $R$ in $\otimes_R$ when it is clear from context that $\otimes$ is a tensor product of chain complexes over $R$-modules.

\begin{definition}[2D Lifted Product Code]
Let $\mc{A}:A_1 \xrightarrow{H_A} A_0$ and $\mc{B}:B_1 \xrightarrow{H_B} B_0$ be classical codes over $R$. A 2D lifted product (LP) code, $\QAB = \mr{LP}(\mc{A},\mc{B})$ is defined as the following product total complex $\mc{Q}_{\mr{AB}} = \mc{A}\otimes_R \mc{B}$:
\[
\begin{tikzcd}[column sep=1em] % tweak 2.2em to taste
Q^{\mathrm{AB}}_0
  & 
  & C_0^A \otimes C_0^B
  & \\
Q^{\mathrm{AB}}_1 \arrow[u, "H^{\mathrm{AB}}_X"]
  & C_1^A \otimes C_0^B \arrow[ru, "H_A \otimes I"]
  &
  & C_0^A \otimes C_1^B \arrow[lu, "I \otimes H_B"'] \\
Q^{\mathrm{AB}}_2 \arrow[u, "(H^{\mathrm{AB}}_Z)^T"]
  &
  & C_1^A \otimes C_1^B
      \arrow[lu, "I \otimes H_B"]
      \arrow[ru, "H_A \otimes I"']
  &
\end{tikzcd}
\]

% \[\begin{tikzcd}
% Q^{\mr{AB}}_0 & \hspace{-2em} & C_0^A \otimes C_0^B  &\hspace{-2em}  \\
% Q^{\mr{AB}}_1 \arrow[u, "H^{\mr{AB}}_X"]  & \hspace{-2em} C_1^A\otimes C_0^B \arrow[ru, "H_A \otimes I"] & &\hspace{-2em}  C_0^A\otimes C_1^B \arrow[lu, "I\otimes H_B"'] \\
% Q^{\mr{AB}}_2 \arrow[u, "(H^{\mr{AB}}_Z)^T"] & \hspace{-2em} & C_1^A\otimes C_1^B \arrow[lu, "I\otimes H_B"] \arrow[ru, "H_A \otimes I"'] &\hspace{-2em}      
% \end{tikzcd}\]
with the following boundary maps:
\begin{align}
\partial_1^{\mr{AB}} & = H_X^{\mr{AB}} =                 \begin{bmatrix}
        H_A \otimes I &I \otimes H_B
    \end{bmatrix},\\
\partial_2^{\mr{AB}} & = (H^{\mr{AB}}_Z)^{T} =            \begin{bmatrix}
        I \otimes H_B \\
        \hline 
        H_A \otimes I
    \end{bmatrix}.
\end{align}
The basis of $Q^{\mr{AB}}_2 = C_1^A\otimes C_1^B$ are the $Z$ checks, $Q^{\mr{AB}}_1 = (C_1^A\otimes C_0^B)\oplus(C_0^A\otimes C_1^B)$ are the qubits, and $Q^{\mr{AB}}_0 = C_0^A\otimes C_0^B$ are the $X$ checks.
\end{definition}

\begin{definition}[3D Lifted Product Code]
The tensor product of a 2D LP code $\QAB$ and another classical (1D) LP code $\mc{C}_C$ is a 3D LP code $\QABC=\mr{LP}(\mr{LP}(\mc{A},\mc{B}),\mc{C})=\mr{LP}(\mc{A},\mc{B},\mc{C})$ defined as the product total complex $\mc{Q}_{\mr{AB}}\otimes_R \mc{C}_C$: 
% \begin{equation}
% \begin{tikzcd}
% Q^{\mr{ABC}}_0 & \hspace{-2.5em}  &  {\scriptstyle C_0^A\otimes C_0^B\otimes C_0^C} & \hspace{-1.5em} \\
% Q^{\mr{ABC}}_1 \arrow[u, "H^{\mr{ABC}}_X"] & \hspace{-2.5em} {\scriptstyle C_1^A\otimes C_0^B\otimes C_0^C}  \arrow[ru]  &  {\scriptstyle C_0^A\otimes C_1^B\otimes C_0^C}  \arrow[u]  & \hspace{-1.5em} {\scriptstyle C_0^A\otimes C_0^B\otimes C_1^C}  \arrow[lu]  \\
% Q^{\mr{ABC}}_2 \arrow[u, "(H^{\mr{ABC}}_Z)^T"] & \hspace{-2.5em} {\scriptstyle C_1^A\otimes C_1^B\otimes C_0^C}  \arrow[u] \arrow[ru] & {\scriptstyle C_1^A\otimes C_0^B\otimes C_1^C}  \arrow[lu] \arrow[ru] & \hspace{-1.5em} {\scriptstyle C_0^A\otimes C_1^B\otimes C_1^C}  \arrow[u] \arrow[lu] \\
% Q^{\mr{ABC}}_3 \arrow[u, "(M^{\mr{ABC}}_Z)^T"] & \hspace{-2.5em} & {\scriptstyle C_1^A\otimes C_1^B\otimes C_1^C} \arrow[lu] \arrow[u] \arrow[ru] & \hspace{-1.5em} 
% \end{tikzcd}
% \label{eq:3D_complex}
% \end{equation}

\begin{equation}
\begin{tikzcd}[column sep=0.2em]
Q^{\mr{ABC}}_0 &  & {\scriptstyle C_0^A\otimes C_0^B\otimes C_0^C} & \\
Q^{\mr{ABC}}_1 \arrow[u, "H^{\mr{ABC}}_X"] & {\scriptstyle C_1^A\otimes C_0^B\otimes C_0^C}  \arrow[ru]  & {\scriptstyle C_0^A\otimes C_1^B\otimes C_0^C}  \arrow[u]  & {\scriptstyle C_0^A\otimes C_0^B\otimes C_1^C}  \arrow[lu]  \\
Q^{\mr{ABC}}_2 \arrow[u, "(H^{\mr{ABC}}_Z)^T"]  & {\scriptstyle C_1^A\otimes C_1^B\otimes C_0^C}  \arrow[u] \arrow[ru] & {\scriptstyle C_1^A\otimes C_0^B\otimes C_1^C}  \arrow[lu] \arrow[ru] & {\scriptstyle C_0^A\otimes C_1^B\otimes C_1^C}  \arrow[u] \arrow[lu] \\
Q^{\mr{ABC}}_3 \arrow[u, "(M^{\mr{ABC}}_Z)^T"] & & {\scriptstyle C_1^A\otimes C_1^B\otimes C_1^C} \arrow[lu] \arrow[u] \arrow[ru] &                                                      
\end{tikzcd}
\label{eq:3D_complex}
\end{equation}
with the boundary maps:
\begin{equation}
    \begin{aligned}
\partial_1^{\mr{ABC}}  & = H^{\mr{ABC}}_X = \begin{bmatrix}
        \partial^{\mr{AB}}_1 \otimes I & I \otimes H_C
    \end{bmatrix} \\
    & = \begin{bmatrix}
        H_A \otimes I\otimes I & I \otimes H_B \otimes I &
        I\otimes I\otimes H_C
    \end{bmatrix},\\
\partial_2^{\mr{ABC}} & = (H^{\mr{ABC}}_Z)^T =          \begin{bmatrix}
        \partial_2^{\mr{AB}}\otimes I & I \otimes H_C \\
        0 & \partial_1^{\mr{AB}} \otimes I
    \end{bmatrix} \\
    & =\begin{bmatrix}
        I \otimes H_B \otimes I & I\otimes I \otimes H_C & 0\\
        H_A \otimes I \otimes I & 0 & I\otimes I \otimes H_C \\
        0& H_A \otimes I \otimes I & I \otimes H_B \otimes I
    \end{bmatrix},\\
\partial_3^{\mr{ABC}} 
    & = (M^{\mr{ABC}}_Z)^{T} 
     = \begin{bmatrix}
        I \otimes H_C \\
        \partial_2^{\mr{AB}} \otimes I
    \end{bmatrix} =\begin{bmatrix}
        I\otimes I \otimes H_C\\
        I\otimes H_B \otimes I\\
        H_A \otimes I \otimes I
    \end{bmatrix},
\end{aligned}
\label{eq:3D_check_mats}
\end{equation}
where $\partial_1^{\mr{ABC}}$, $\partial_2^{\mr{ABC}}$, and $\partial_3^{\mr{ABC}}$ give the X-check matrix $H_X^{\mr{ABC}}$, the transpose of the Z-check matrix $H_Z^{\mr{ABC}}$, and the transpose of the meta-Z check matrix $M_Z^{\mr{ABC}}$, respectively.
    % Let $\mathcal{C}^A,\mathcal{C}^B$, and $\mathcal{C}^C$ be three classical codes with check matrices $c^*\in R^{r_A\times n_A}$, $b\in R^{r_B\times n_B}$, and $c^*\in R^{r_C\times n_C}$, respectively (the binary matrix after lift is not necessarily full rank). A 3D lifted product code is constructed from the total complex of the tensor product of $\mathcal{C}^A,\mathcal{C}^B$, and $\mathcal{C}^C$. 

\end{definition}
Note that we have directly expressed the $3$-chain complex $\mc{Q}_{\mr{ABC}}$ as the product total complex of $\mc{C}_A \otimes_R \mc{C}_B \otimes_R \mc{C}_C$~\cite{breuckmann_quantum_2021}. We could equivalent write it as the product total complex of $\mc{Q}^{\mr{AB}}\otimes \mc{C}_C$, $\mc{Q}^{\mr{AC}}\otimes \mc{C}_B$, or $\mc{Q}^{\mr{BC}}\otimes \mc{C}_A$, following the definition in Def.~\ref{def:product_complex}.
% Assigning $Q_0$, $Q_1$, $Q_2$, and $Q_3$ to the X-checks, qubits, Z-checks and meta-Z checks, respectively, we can construct the check matrices $H_X$, $H_Z^T$, and $M_Z$ as
% \begin{align}
%     H_X & = \renewcommand\arraystretch{1.5}
% \mleft[\begin{array}{c|c|c}
%    c\otimes I_{r_B}\otimes I_{r_C}  &  I_{r_A}\otimes b\otimes I_{r_C} & I_{r_A}\otimes I_{r_B}\otimes c
% \end{array}
% \mright]\\
%     H_Z^T & = \renewcommand\arraystretch{1.5}
% \mleft[
% \begin{array}{c|c|c}
%   I_{n_A}\otimes b\otimes I_{r_C} & I_{n_A}\otimes I_{r_B}\otimes c & \mathbf{0}_{n_Ar_Br_C,r_An_Bn_C}\\
%   \hline
%   c\otimes I_{n_B}\otimes I_{r_C} & \mathbf{0}_{r_An_Br_C,n_Ar_Bn_C} & I_{r_A}\otimes I_{n_B}\otimes c\\
%   \hline
%   \mathbf{0}_{r_Ar_Bn_C,n_An_Br_C} & c\otimes I_{r_B}\otimes I_{n_C} & I_{r_A}\otimes b\otimes I_{n_C}
% \end{array}
% \mright]\\
%     M_Z & = \renewcommand\arraystretch{1.5}
% \mleft[
% \begin{array}{c}
%   I_{n_A}\otimes I_{n_B}\otimes c\\
%   \hline
%   I_{n_A}\otimes b\otimes I_{n_C}\\
%   \hline
%   c\otimes I_{n_B}\otimes I_{n_C}
% \end{array}
% \mright]
% \end{align}
% \end{definition}
% \begin{definition}[Lifted total complex]
% \end{definition}
\begin{definition}[Chain map]
    Given two $D$-chain complexes $\mc{A}=\{\{A_i\}_i,\{\partial^{A}_i\}_i\}$ and $\mc{B}=\{\{B_i\}_i,\{\partial^{B}_{i}\}_i\}$ over $R$, a chain map $\gamma:\mc{A}\rightarrow\mc{B}$ is a collection of $R$-linear maps $\{\gamma_i:A_i\rightarrow B_i\}_{i\in [D]}$ such for each $i\in[D]$, $\partial^{B}_i\gamma_i = \gamma_{i-1}\partial^{A}_i$; i.e., the following diagram commutes.
    \begin{equation}
        \begin{tikzcd}
{B_i} \arrow[r, "\partial^{B}_i"]                 & {B_{i-1}}                    \\
{A_i} \arrow[u, "\gamma_i"] \arrow[r, "\partial^{A}_i"] & {A_{i-1}} \arrow[u, "\gamma_{i-1}"]
\end{tikzcd}.
    \end{equation}
\end{definition}
\begin{definition}[Chain map between two 3D LP quantum codes]
    Let $\tilde{\mathcal{Q}}_{\mr{AB}}$ and $\mathcal{Q}_{\mr{ABC}}$ be 3-chain complexes over $R$, each representing a 3D LP code. A chain map $\gamma: \tilde{\mathcal{Q}}_{\mr{AB}}\rightarrow \mathcal{Q}_{\mr{ABC}}$ consists of $\gamma = \{\gamma_{0},\gamma_1,\gamma_2,\gamma_3\}$ such that the following diagram commutes: 
\begin{equation}\begin{tikzcd}[column sep=1.5em]
\mathcal{Q}_{\mr{ABC}}: &\hspace{-1em} Q_3^{\mr{ABC}} \arrow[r, "{\partial_3^{\mr{ABC}}}"]                              & Q_2^{\mr{ABC}} \arrow[r, "{\partial_2^{\mr{ABC}}}"]                           & Q_1^{ABC} \arrow[r, "{\partial_1^{\mr{ABC}}}"]                           & Q_0^{ABC}                          \\
\tilde{\mathcal{Q}}_{\mr{AB}}: &\hspace{-1em} \tilde{Q}^{\mr{AB}}_3 \arrow[r, "{\tilde{\partial}_3^{\mr{AB}}}"] \arrow[u, "\gamma_3"] & \tilde{Q}^{\mr{AB}}_2 \arrow[r, "{\tilde{\partial}_2^{\mr{AB}}}"] \arrow[u, "\gamma_2"] & \tilde{Q}^{\mr{AB}}_1 \arrow[r, "{\tilde{\partial}_1^{\mr{AB}}}"] \arrow[u, "\gamma_1"] & \tilde{Q}^{\mr{AB}}_0 \arrow[u, "\gamma_0"]
\end{tikzcd}
\end{equation}
\end{definition}

\begin{lemma}
\label{lemma:binary_chain_map}
    Let $\mc{A},\mc{B}$ be two $D$-chain complexes over $R$s. If $\gamma:\mc{A}\rightarrow \mc{B}$ is a chain map over $R$, then $\gamma = \{\gamma_i\}_{i \in [D]}$, a collection of $R$-linear maps, induces a chain map on the binary representation of the codes, $\mbb{B}(\gamma):\mbb{B}(\mc{A})\rightarrow \mbb{B}(\mc{B})$, where $\mbb{B}(\gamma) = \{\mbb{B}(\gamma_i)\}_{i \in [D]}$ is a collection of $\mbb{F}_2$-linear maps given by the binary representations of $\{\gamma_i\}$. 
\end{lemma}
\begin{proof}
    In order for $\mbb{B}(\mbb{\gamma})$ to be a valid chain map from $\mbb{B}(\mbb{\mc{A}})$ to $\mbb{B}(\mbb{\mc{B}})$, its suffices to show that $\mbb{B}(\partial_i^B)\mbb{B}(\gamma_i) = \mbb{B}(\gamma_{i-1})\mbb{B}(\partial_i^A)$ for each $i \in \{1,\ldots,D\}$, i.e. the following diagram commutes:
    \begin{equation}
        \begin{tikzcd}[column sep=3.5em]
\mbb{B}({B_i}) \arrow[r, "\mbb{B}(\partial^{B}_i)"]                 & \mbb{B}({B_{i-1}})                    \\
\mbb{B}({A_i}) \arrow[u, "\mbb{B}(\gamma_i)"] \arrow[r, "\mbb{B}(\partial^{A}_i)"] & {\mbb{B}(A_{i-1})} \arrow[u, "\mbb{B}(\gamma_{i-1})"]
\end{tikzcd}
    \end{equation}
    Since $\gamma$ is a chain map, $\partial_i^B\gamma_i = \gamma_{i-1}\partial_i^A$. Applying the binary lift operator $\mbb{B}(\cdot)$, which maps each  matrix entry in $g\in R$ to its faithful binary representation as a $l\times l$ matrix, we get $\mbb{B}(\partial_i^B\gamma_i) = \mbb{B}(\gamma_{i-1}\partial_i^A)$. Since $\mbb{B}(\cdot)$ is a linear map, this implies $\mbb{B}(\partial_i^B)\mbb{B}(\gamma_{i})=\mbb{B}(\gamma_{i - 1})\mbb{B}(\partial_i^A)$. Thus, the binary map $\mbb{B}(\gamma)$ is a valid chain map from $\mbb{B}(\mc{A})$ to $\mbb{B}(\mc{B})$. 
\end{proof}
\section{Proof of Lemma 1 - Homomorphic CNOT \label{sec:appx_proofhomomorphiccnot}}
\begin{lemma}[Homomorphic CNOT (Lemma 1 of the main text)]
    For any 3D LP code $\QABC=\mathrm{LP}(\mc{C}_A,\mc{C}_B,\mc{C}_C)$, there exists a logical CNOT circuit controlled by $\QABC$ and targeting any of its 2D component codes, implemented entirely via transversal physical CNOTs.
    \label{lemma:homo_CNOT_SI}
\end{lemma}
We will show: (1) there exists such a logical CNOT circuit, and (2) it is implemented entirely via transversal physical CNOTs. \\\\
\textbf{Note:}\begin{itemize}
    \item The argument presented below applies to homomorphic CNOTs between $\QABC$ and any of its component 2D codes, $\{\QAB, \QAC, \QBC\}$, by symmetry. We chose to work with $\QAB$ without loss of generality. 
    \item This homomorphic CNOT construction can be generalized to a ``parent" LP code at an arbitrary dimension $D$ and one of its $D-1$-dimensional component codes. 
    In particular, it is possible to construct a chain map and thus logical CNOT between any $D$-dimensional LP code and any of its $D-1$ component code. 
    Since the 2D-3D code pair are the lowest-dimension pair that can achieve universal computation, we focus on this case to make the construction fully concrete. 
\end{itemize}
\textbf{(1) Construction of a logical CNOT:} Recall that a logical CNOT controlled by $\QABC$ targeting a 2D component code $\QAB$ is defined by a chain map between the binary representations of the codes, $\mbb{B}(\tQAB)\rightarrow \mbb{B}(\QABC)$. Specifically, it is implemented at the physical level by $\gamma_1$, which maps between the qubits of the codes. Since there is a faithful mapping from elements in $R$ to its binary representation, and a chain map over $R$ induces a chain maps over $\mbb{F}_2$, we can equivalently find a chain map over $R$, $\gamma:\QAB\rightarrow\QABC$. Thus, to show that a logical CNOT circuit exists, we show that it is always possible to construct a chain map $\gamma:\QAB\rightarrow \QABC$.\\\\
\noindent
\textbf{Claim 1.1:}
    A chain map between two LP codes can be constructed as the tensor product of chain maps between their base codes~\cite{xu2025fast}. Specifically, let $\mc{Q}_{\mr{XY}}=\mr{LP}(X,Y)$ and $\mc{Q}_{\mr{XY}}'=\mr{LP}(X',Y')$ be 3-chain complexes over $R$ with boundary maps $\{\partial^{XY}_i\}_{i \in [3]}$ and $\{\partial_i^{{XY}^{\prime}}\}_{i \in [3]}$, respectively, where $X,X'$ are 2-chain complexes of 2D LP codes and $ Y,Y'$ are 1-chain complexes of classical codes. Let $\gamma_X =\{\gamma^X_i\}_{i = 0}^2:X'\rightarrow X$ and $\gamma_Y = \{\gamma^Y_i\}_{i = 0}^1:Y'\rightarrow Y$ be chain maps on the 2D quantum codes and the classical codes, respectively. Then, $\gamma_{XY} = \{\gamma^{XY}_i\} _{i = 0}^3 =  \gamma_X\otimes_R \gamma_Y$, defined as
    \begin{equation}
    \begin{aligned}
         \gamma_0 & = \gamma_0^X\otimes\gamma_0^Y,\\
    \gamma_1 & = (\gamma_1^X \otimes \gamma_0^Y)\oplus (\gamma_0^X\otimes\gamma_1^Y), \\
    \gamma_2 & = (\gamma_2^X\otimes \gamma_0^Y)\oplus (\gamma_1^X\otimes \gamma_1^Y),\\
        \gamma_3 & = \gamma_2^X\otimes\gamma_1^Y,
    \end{aligned}
    \end{equation}
    induces a chain map $\gamma_{XY}:\mc{Q}_{XY}'\rightarrow \mc{Q}_{XY}$ on the 3D LP codes.
\begin{proof}
To show that $\gamma_X\otimes_R \gamma_Y$ induces a chain map from $\mc{Q}_{XY}'$ to $\mc{Q}_{XY}$, we need to show that, for each degree $k\in\{1,2,3\}$, $\partial_k^{XY}\gamma^{XY}_k =\gamma^{XY}_{k-1}{\partial_k^{XY}}'$. 
% $\gamma_X\otimes\gamma_Y$ is defined as:
% \begin{align}
% \gamma_0 & = \gamma_0^X\otimes\gamma_0^Y\\
% \gamma_1 & = (\gamma_1^X \otimes \gamma_0^Y)\oplus (\gamma_0^X\otimes\gamma_1^Y) \\
% \gamma_2 & = (\gamma_2^X\otimes \gamma_0^Y)\oplus (\gamma_1^X\otimes \gamma_1^Y)\\
%     \gamma_3 & = \gamma_2^X\otimes\gamma_1^Y
% \end{align}
Since $X$ is a 2-chain complex and $Y$ is a 1-chain complex, the chain maps $\gamma_X$ and $\gamma_Y$ satisfy $\partial_i^X\gamma_i^X = \gamma_{i-1}^X {\partial_i^X}'$ for $i = 1, 2$ and $\partial_1^Y\gamma_1^Y = \gamma_{0}^Y {\partial_1^Y}'$. Therefore, for each degree $k$, we check:
\begin{equation}
    \begin{aligned}
    k=1:\;  \partial_1^{XY}\gamma_1 & = \begin{bmatrix}
        \partial^{X}_1 \otimes I & I \otimes \partial_1^{Y}
    \end{bmatrix}\begin{bmatrix}
        \gamma_1^X\otimes\gamma_0^Y & 0 \\
       0 &  \gamma_0^X\otimes \gamma_1^Y
    \end{bmatrix} \\
    & = \begin{bmatrix}
        \partial_1^X\gamma_1^X\otimes \gamma_0^Y & \gamma_0^X\otimes \partial_1^Y\gamma_1^Y
    \end{bmatrix}\\
    & = \begin{bmatrix}
        \gamma_0^X{\partial_1^X}'\otimes \gamma_0^Y & \gamma_0^X\otimes \gamma_0^Y{\partial_1^Y}'
    \end{bmatrix} \\
    % \tag{$\gamma_k^X\partial_k^X = \partial_{k-1}^X{\gamma_k^X}'$ and same for $Y$}\\
    & = (\gamma_0^X\otimes\gamma_0^Y)\begin{bmatrix}
        {\partial_1^X}'\otimes I & I \otimes {\partial_1^Y}'
    \end{bmatrix} \\
    & = \gamma_0{\partial_1^{XY}}'
    \end{aligned}
\end{equation}
\begin{equation}
    \begin{aligned}
        k=2:\;\partial_2^{XY} \gamma_2 & = \begin{bmatrix}
        \partial_2^{X}\otimes I & I \otimes \partial_1^Y \\
        0 & \partial_1^{X}\otimes I
    \end{bmatrix}\begin{bmatrix}
        \gamma_2^X\otimes\gamma_0^Y &0\\
        0 & \gamma_1^X\otimes \gamma_1^Y
    \end{bmatrix}\\
    & = \begin{bmatrix}
        \partial_2^X\gamma_2^X\otimes\gamma_0^Y & \gamma_1^X\otimes \partial_1^Y\gamma_1^Y\\
        0 & \partial_1^X\gamma_1^X\otimes\gamma_1^Y
    \end{bmatrix}\\
    & = \begin{bmatrix}
        \gamma_1^X{\partial_2^X}'\otimes\gamma_0^Y & \gamma_1^X\otimes\gamma_0^Y{\partial_1^Y}'\\
        0 & \gamma_0^X{\partial_1^X}'\otimes\gamma_1^Y
    \end{bmatrix}\\
    & = \begin{bmatrix}
        \gamma_1^X\otimes\gamma_0^Y & 0 \\
        0 & \gamma_0^X\otimes\gamma_1^Y
    \end{bmatrix}\begin{bmatrix}
        {\partial_2^{X}}'\otimes I & I \otimes {\partial_1^Y}' \\
        0 & {\partial_1^{X}}'\otimes I
    \end{bmatrix} \\
    & = \gamma_1{\partial_2^{XY}}'\\
    \end{aligned}
\end{equation}
\begin{equation}
    \begin{aligned}
        k=3:\;\partial_3^{XY} \gamma_3 & = \begin{bmatrix}
        I \otimes \partial_1^Y\\
        \partial_2^{X}\otimes I
    \end{bmatrix} (\gamma_2^X\otimes \gamma_1^Y) \\
    & = \begin{bmatrix}
        \gamma_2^X\otimes \partial_1^Y\gamma_1^Y\\
        \partial_2^X\gamma_2^X\otimes\gamma_1^Y
    \end{bmatrix}\\
    & = \begin{bmatrix}
        \gamma_2^X\otimes\gamma_0^Y{\partial_1^Y}'\\
        \gamma_1^X{\partial_2^X}'\otimes\gamma_1^Y
    \end{bmatrix}\\
    & = \begin{bmatrix}
        \gamma_2^X\otimes\gamma_0^Y &0\\
        0 & \gamma_1^X\otimes \gamma_1^Y
    \end{bmatrix}\begin{bmatrix}
        I \otimes {\partial_1^Y}'\\
        {\partial_2^{X}}'\otimes I
    \end{bmatrix}\\
    & = \gamma_2{\partial_3^{XY}}'.
    \end{aligned}
\end{equation}
    Therefore, $\gamma_{XY} = \gamma_X\otimes \gamma_Y$ is a valid chain map for ${\mc{Q}_{XY}}'\rightarrow\mc{Q}_{XY}$.
\end{proof}
\noindent \textbf{Claim 1.2:} Given any 3D LP code $\QABC$ and a 2D component codes $\QAB$, we can always construct a chain map $\gamma:\QAB\rightarrow\QABC$.
\begin{proof}
    The 2D code $\QAB$ is isomorphic to $\tQAB = \mr{LP}(\QAB,\mc{C}_0)$, where $\mc{C}_0$ is the trivial subcomplex of $\mc{C}_C$, with chain complex $0
    \rightarrow R^1$. Let $X'=X=\QAB$, $Y = \mc{C}_{\mr{C}}$, and $Y'=\mc{C}_0$. Then, $\QAB\simeq\tQAB=\mr{LP}(X',Y')$ and $\QABC = \mr{LP}(X,Y)$, so $\gamma:\tQAB\rightarrow\QABC = \mr{LP}(X',Y')\rightarrow \mr{LP}(X,Y)$. Since $X'=X$, we can simply set $\gamma_X:X'\rightarrow X$ to be the identity map --- i.e., identity matrices as each component. The boundary operator of $Y$, $\partial_1^Y$ is a $r_Y\times n_Y$ matrix over $R$. The boundary operator of $Y'$ is the empty $1\times 0$ matrix over $R$, i.e. $\partial_1^{Y \prime} = []_{1\times 0}$.
    For $\gamma_Y:Y' = \mc{C}_0 \rightarrow Y=\mc{C}_C$, set $\gamma_1^Y \coloneq [\;]_{n_Y\times 0}$ and $\gamma_0^Y \coloneq \mb{e}^{r_Y}_q$, the unit vector in $R^{r_Y}$ with the identity element $e\in R$ at row $q \in [r_Y]$. This forms a valid chain map $\gamma_Y:Y'\rightarrow Y$, since $\partial_1^Y\gamma_1^Y = \gamma_0^Y{\partial_1^Y}'=[\;]_{r_Y\times 0}$.
    Note that this classical code chain map $\gamma_Y$ can be understood using the general framework in Ref.~\cite{xu2025fast}: $\mc{C}_0$ is obtained by shortening on the transpose of the check matrix $\partial^Y_1$ on all but one bit (in $R$) of $Y$ ($\mc{C}_C$), which is a structure preserving operation, inducing the homomorphism $\gamma_Y$.

Thus, by taking the product of chain maps for $X$ and $Y$, we obtain a chain map $\gamma = \gamma_{XY}:\mr{LP}(X',Y')\rightarrow\mr{LP}(X,Y)=\tilde{\mc{Q}}_{AB}\rightarrow\QABC$. In particular, $\gamma_1$, which maps between the qubits of the codes, takes the form:
\begin{align}
    \gamma_1 & = (\gamma_1^{X} \otimes \gamma_0^{Y})\oplus (\gamma_0^{X}\otimes\gamma_1^{Y})\\
    & = (I_{n_{\mr{AB}}}\otimes \mb{e}^{r_Y}_{q})\oplus(I_{r_{\mr{AB}}^X}\otimes [\;]_{n_C\times 0})\\
    & = (I_{n_{\mr{AB}}} \otimes \mb{e}^{r_C}_{q})\oplus([\;]_{r_{\mr{AB}}^Xn_C\times 0}) \\
    & = \begin{bmatrix}
        I_{n_{\mr{AB}}}\otimes \mb{e}^{r_C}_q\\
        0_{r_{\mr{AB}}^Xn_C\times n_{AB}}
    \end{bmatrix},
    \label{eq:gamma_1}
\end{align}
where $n_{\mr{AB}}$, $n_C$, $r_Y = r_C$ and $r^X_{\mr{AB}}$ denote the dimension of the qubits of $\mc{Q}_{\mr{AB}}$ ($Q^{\mr{ABC}}_1$), the classical bits of $\mc{C}_C$ $(C_1)$, the classical checks of $\mc{C}_C$ ($C_0$), and the $X$ checks of $\mc{Q}_{\mr{AB}}$ ($Q^{\mr{AB}}_0$), respectively, over $R$. 
The choice of the row $q \in [r_C]$ can be arbitrary. 
\end{proof}

\noindent \textbf{Claim 1.3:} The above construction for $\gamma_1$ defines a logical CNOT between $\tQAB$ and $\QABC$, and it is an inclusion map on the physical qubits.
\begin{proof}
Since $\gamma$ is chain map from $\tQAB$ to $\QABC$, it maps boundaries to boundaries, and cycles to cycles~\cite{huang2022homomorphic, xu2025fast}. 
Therefore, $\gamma_1$ induces a chain map between the first homologies $H_1$ of the two codes \cite{hatcher2005algebraic}, $\bar{\gamma}_1:H_1(\tQAB)\rightarrow H_1(\QABC)$. The induced map is defined as $\bar{\gamma}_1([v]) = [\gamma_1 (v)]$ for any $v\in Z_1(\tQAB)$ --- i.e., logical Z cosets of the 2D code $\QAB$ (or more specifically, of $\tQAB=\QAB\otimes\mc{C}_0$ where $\mc{C}_0$ is the trivial classical code with zero bits and one check) are mapped to logical Z cosets of the 3D code $\QABC$. 
Clearly, according to Eq.~\eqref{eq:gamma_1}, $\gamma_1$ is an inclusion map when viewing the qubits of $\tilde{Q}_{\mr{AB}}$, $\tilde{Q}^{\mr{AB}}_1$, as a subset of those of $\QABC$, $Q^{\mr{ABC}}_1$, and different choice of $q$ in Eq.~\eqref{eq:gamma_1} corresponds to different ways of defining the inclusion map.
% \begin{align}
%     \gamma_1 & = (\gamma_1^{AB} \otimes \gamma_0^{C})\oplus (\gamma_0^{AB}\otimes\gamma_1^{C})\\
%     & = (I_{n_{AB}}\otimes \mb{e}_{i})\oplus(I_{n_{AB,X}}\otimes [\;]_{n_C\times 0})\\
%     & = (I_{n_{AB}}\otimes \mb{e}_{i})\oplus([\;]_{n_{AB,X}n_C\times 0})
% \end{align}
% \begin{align}
%     \gamma_1 & = (\gamma_1^A\otimes\gamma_0^B\otimes\gamma_0^C)\oplus( \gamma_0^A\otimes\gamma_1^B\otimes\gamma_0^C)\oplus(\gamma_0^A\otimes\gamma_0^B\otimes\gamma_1^C)\\
%     & = ([\;]_{n_A\times 0}\otimes I_{r_B}\otimes I_{r_C})\oplus(\mathbf{e}_{i}\otimes I_{n_B}\otimes I_{r_C})\oplus(\mathbf{e}_{i}\otimes I_{r_B}\otimes I_{n_C})\\
%     & = [\;]_{n_Ar_Br_C,0}\oplus(\mathbf{e}_{i}\otimes I_{n_B}\otimes I_{r_C})\oplus(\mathbf{e}_{i}\otimes I_{r_B}\otimes I_{n_C})
% \end{align}
% The first block, $\gamma_1^{AB}\otimes \gamma_0^{C}$ is an inclusion map, and the second block is trivially an inclusion map. Thus, their direct sum, $\gamma_1$ is an inclusion map from the $n_{AB}$ physical qubits of the 2D code $\QAB$ to a corresponding subset of $n_{AB}$ physical qubits of the 3D code $\QABC$.
\end{proof}
Combining the above, this shows there is always exists a logical CNOT circuit from the 3D LP code $\QABC$ to a component 2D code $\QAB$, implemented using transversal physical CNOTs. \\\\
%     Thus we define the inclusion map between the quantum codes $\gamma:\QAB\rightarrow\QABC$ to be the chain map constructed via a product of the classical code homomorphisms $\gamma^A$, $\gamma^B$, and $\gamma^C$. For each layer $k$ of the product chain complex, 
% \begin{align*}
%     \gamma_k\coloneq \bigoplus_{|x_i|=k}\bigg[\bigotimes_{R,i\in\{A,B,C\}}\gamma_{x_i}^i\bigg]
% \end{align*}
% In particular, the layer of the homomorphism which defines the homomorphic CNOT is
% \begin{align}
%     \gamma_1 & = (\gamma_1^A\otimes_R\gamma_0^B\otimes_R\gamma_0^C)\oplus( \gamma_0^A\otimes_R\gamma_1^B\otimes_R\gamma_0^C)\oplus(\gamma_0^A\otimes_R\gamma_0^B\otimes_R\gamma_1^C)
% \end{align}
% The other layers are:
% \begin{align}
%     \gamma_0 & = \gamma_0^A\otimes_R \gamma_0^B \otimes_R \gamma_0^C
% \end{align}
% \begin{align}
%     \gamma_2 & = (\gamma_1^A\otimes_R\gamma_1^B\otimes_R\gamma_0^C)\oplus (\gamma_1^A\otimes_R \gamma_0^B\otimes_R \gamma_1^C)\oplus(\gamma_0^A\otimes_R \gamma_1^B\otimes_R \gamma_1^C)
% \end{align}
% This $\gamma_1$ corresponds to the physical qubits $Q_1'\simeq R^{n_{2D}}$ and $Q_1\simeq R^{n_{3D}}$. 
\noindent \textbf{(2) Transversal physical implementation of the logical CNOT:} Here, we show that the logical CNOT circuit associated with the chain map $\gamma$ is implemented entirely via transversal physical CNOTs.
\begin{proof}
    Since a homomorphic CNOT between $\QABC$ and $\QAB$ is defined by the nonzero entry indices of $\mathbb{B}(\gamma_1)$, where $\gamma_1$ is the degree-1 component of the chain map $\gamma:{\QAB}\rightarrow {\QABC}$, the physical CNOT circuit being transversal means that the weight of the rows and columns of $\mathbb{B}[\gamma_1]$ is at most 1, which clearly holds for Eq.~\eqref{eq:gamma_1}. 
    % Using the product construction of $\gamma$ from above, each column of $\gamma_1$ has exactly one nonzero entry, $e$, and each row has at most one nonzero entry. When we take the embedding into binary as $l\times l$ matrices over $\mathbb{F}_2$, the unit $e$ gets mapped to $I_l$, the binary identity matrix of dimension $l$, which has exactly one 1 in each row and column. Thus, binary encoding of the homomorphism $\mathbb{B}[\gamma_1]$ has exactly one 1 in each column and at most one 1 in each row. At the physical level, $\gamma_1[i,j]=1$ corresponds to applying CNOT gate from the $i$-th qubit of the 3D code $\mathcal{Q}_{3D}$ to the $j$-th qubit of the 2D code $\mathcal{Q}_{2D}$. Since there is at most one 1 in each row and column of $\gamma_1$, no physical qubits in either the 2D or 3D code block is involved in more than one physical CNOT. Therefore it is possible to apply all physical CNOT gates in parallel, forming a transversal implementation of the logical CNOT.
\end{proof}
    % [TODO] analogous to prop. 6 and 7 of \cite{xu_fast_2025}
% It remains to show that $\gamma_{H_1}$ is an inclusion map. The elements on $H_1(\mathcal{Q}_{2D})=\text{ker}(H_X')/\text{im}({H_Z^T}')$ are equivalence classes of Z logical operators. 

% a) This means that it must map cycles, elements of $\text{ker} H_X$ for the 2D code to the 3D code. Suppose $u\in \text{ker} H_X$ of the 2D code (it is a 2D Z-logical), so $H_X u = 0$. Since $\gamma \cdot$ is a chain map, $H_X(\gamma_1 u) = \gamma_0 {H_Z^T}' u=0$.
% b) It must preserve equivalence class structure. 
% $\gamma_{H_1}([v]) = [\gamma_1(v)]$

\section{Proof of Theorem 2 - Transversal dimension jump\label{sec:appx_prooftransversaldimjump}}
\begin{theorem}[Transversal dimension jump (Theorem 2 of the main text)]
    Given any 3D LP code $\QABC = \mr{LP}(\mc{C}_A, \mc{C}_B, \mc{C}_C)$, one can teleport a subset of its logical qubits to or from a 2D component code in parallel using the homomorphic CNOT of Lemma~\ref{lemma:homo_CNOT_SI}, provided they are logically transversal.
    \label{th:transversal_dim_jump}
\end{theorem}
\begin{proof}
Let $\gamma = \{\gamma_0,\gamma_1,\gamma_2,\gamma_3\}$ be a chain map between a 3D LP code $\QABC$ and $\tQAB\simeq\QAB$, one of its component codes, where $\gamma_1$ is given by Eq.~\eqref{eq:gamma_1}. Let $\bar{\gamma}_1:H_1(\tQAB)\rightarrow H_1(\QABC)$ be the induced map between the first homologies, which maps between cosets of logical Z operators of $\tQAB$ and $\QABC$. As described in the main text, if $\bar{\gamma}_1$ is injective, then, setting it in a certain logical basis, it can be written in the following transversal form:
\begin{equation}
    \bar{\gamma}_1=\begin{pmatrix} I_{k_{\mr{AB}}} \\ 0_{(k_{\mr{ABC}} - k_{\mr{AB}})\times k_{\mr{AB}}}\end{pmatrix},\qquad (k_{\mr{AB}}\leq k_{\mr{ABC}}),
\end{equation}
where $k_{\mr{AB}}$ and $k_{\mr{ABC}}$ denotes the dimension of $H_1(\tQAB)$ and $H_1(\mc{Q}_{\mr{ABC}})$, respectively.
In particular, injectivity of $\bar{\gamma}_1$ --- $\ker{\bar{\gamma}_1}=0$ --- means that no nontrivial logical Z operator of the 2D code is mapped to a stabilizer in the 3D code, and every logical qubit of the 2D code couples distinctly to one of the 3D code. Therefore, the homomorphic CNOT circuit is \emph{logically transversal} (not to be confused with transversal physical implementation).
By combining such a transversal logical CNOT circuit with the logical teleportation circuit of Fig.1(b–c) of the main text, we can teleport between all logical qubits are $\tQAB$ and the corresponding subset of $k_\mr{AB}$ logical qubits $\QABC$ in parallel. 
\end{proof}

\begin{corollary}
    Given $N$ copies of component 2D codes $\mc{Q}_{2D}$, $\{\mc{Q}_{2D}^{(i)}:\mc{Q}_{2D}^{(i)}\in\{\QAB,\QBC, \QAC\}\}_{i \in [N]}$, each coupled to $\mc{Q}_{3D}$ with a homomorphism $\gamma_1^{(i)}$, we can implement the homomorphic CNOTs associated with $\{\gamma_1^{(i)}\}_{i \in [N]}$ in parallel if the images of $\{\gamma_1^{(i)}\}_{i \in [N]}$ are disjoint. 
\end{corollary}
% \begin{proof}
% Let $\{\mc{Q}_{2D}^{(i)}\}_{i \in [N]}$ be $N$ copies of component 2D codes for a given 3D code $\mc{Q}_{3D}$ such that the images of $\{\gamma_i\}_{i \in [N]}$ are all mutually disjoint (i.e., the indices of the nonzero rows of the $\gamma_i$'s do not overlap). The control qubits on the 3D code block of the physical implementation of a homomorphic CNOT $\mc{Q}_{2D}^{(i)}$ is determined by the row indices of the matrix $\gamma_i$. Since all the row indices are nonoverlapping, the control qubits are not overlapping. The target qubits lie in the different 2D code blocks. Therefore, we can apply homomorphic CNOTs in parallel from a 3D code to any number of its component 2D codes provided that the $\gamma_1^{(i)}$'s have disjoint images.
% \end{proof}

\subsection{Inclusion map and logical transversality\label{sec:appx_inclusionmaplogicaltransversality}}
In this section, we describe in details the conditions under which the induced homologies map $\bar{\gamma}_1: H_1(\tilde{Q}_{\mr{AB}}) \rightarrow H_1(\QABC)$ is injective, which indicates that the logical CNOTs are also transversal.
% Denote $\tilde{\mc{Q}}^{BC}_1 := R^1 \otimes_R Q^{BC}_1$. We have $\tilde{\mc{Q}}^{BC}_1 \simeq Q^{BC}_1$ using the fact $R^m \otimes_{R} R^n \simeq R^{mn}$. Treating $R^1 \subset C^A_0 = R^{r_A}$, we have the phyiscal inclusion map
% \begin{equation}
%     \gamma_1: \tilde{\mc{Q}}^{BC}_1 \hookrightarrow Q^{ABC}_1,
% \end{equation}
% where $\gamma_1 = e^{r_A}_1 \otimes I_{n_{BC}}$. 

% From the standard homological algebra, $\gamma_1$ induces a module-homomorphism on the homologies:
% \begin{equation}
%     \bar{\gamma}_1: H_1(\mc{Q}^{BC}) \rightarrow H_1(\mc{Q}^{ABC}),
% \end{equation}
% via $\bar{\gamma_1}([l]_{BC}) = [\gamma_1(l)]_{ABC}$. More specifically, we have that for any $1$-cycle of $\mc{Q}^{BC}$, $l \in Z_1(\mc{Q}^{BC})$, 
% \begin{equation}
%     \gamma_1(l + \partial_2^{BC}(s_z)) = \gamma_1(l) + \partial_2^{ABC}(s_z^{\prime}),
% \end{equation}
% where $s_z \in Q_2^{BC}$ and $s_z^{\prime} \in Q_2^{ABC}$, and $\gamma_1(l)$ is a $1$-cycle of $\mc{Q}^{ABC}$, i.e. $l \in Z_1(\mc{Q}^{ABC})$. 

% However, in general, $\bar{\gamma}_1$ is not guaranteed to be injective, i.e. $[\gamma_1(l)]_{BAC} \neq [\gamma_1(l^{\prime})]$ for $[l]_{BC} \neq [l]_{ABC}$. 
% We know that a linear map between modules is injective iff its kernel is trivial. As such, we have the following condition on $\bar{\gamma}_1$ being injective:
Recall that $H_1=Z_1/B_1=\ker{H_X}/\im{H_Z^T}$ of a quantum code corresponds to its equivalence classes of logical Z operators.
\begin{proposition}
     $\bar{\gamma}_1$ is injective if and only if $\gamma_1$ satisfies that: for all $[l]_{\mr{AB}} \in H_1(\QAB)$, $\gamma_1(l) \in B_1(\QABC)$ if and only if $[l]_{\mr{AB}}$ is trivial, i.e. $l \in B_1(\QAB)$.
    \label{prop:logical_injective}
\end{proposition}
\begin{proof}
    This follows from the definition of the induced chain map and injectivity. A logical Z operator $l$ (possibly a trivial operator, i.e., a Z stabilizer) of the 2D code is mapped a Z stabilizer of the 3D code if and only if the $l$ itself is a Z stabilizer of the 2D code. This ensures that $\ker{\bar{\gamma}_1}=\{[0]\}=0$.
\end{proof}

Note that Proposition~\ref{prop:logical_injective} is a condition that one would easily check numerically. We have numerically verified that it holds for all the univariate LP codes (the lifted Toric codes) as well as the trivariate tricycle code in Table I of the main text.

As for more concrete algebraic conditions, we have shown in the main text that $\bar{\gamma}_1$ is injective for any 3D hypergraph product (HGP) code and any of its 2D component codes, i.e. Proposition~\ref{prop:logical_injective} is trivially satisfied for all HGP code pairs.
Here, we derive a simple algebraic condition for the multivariate tricycle-bicycle code pairs.

For an odd-order group $G$, its group algebra ring $R = \mbb{F}_2[G]$ is semi-simple and we have the Künneth formula for $H_1(\QABC)$:
\begin{equation}
    H_1(\QABC) = H_1(\QAB)\otimes H_0(\mc{C}_C) \oplus H_0(\QAB)\otimes H_1(\mc{C}_C),
    \label{eq:HGP_kunneth}
\end{equation}
where $H_1(\mc{C}_C)$ (resp. $H_0(\mc{C}_C)$) denotes the kernel (resp. cokernel) of $H_C \in R^{r_C\times n_C}$ --- the check matrix of $\mc{C}_C$. 
In this case, $\bar{\gamma}_1$ takes the form:
\begin{equation}
    \bar{\gamma}_1: R^1 \otimes_R H_1(\tQAB) \rightarrow H_0(\mc{C}^A)\otimes_R H_1(\mc{Q}^{BC}).
\end{equation}
For tricycle/bicycle codes, we have $H_A = a, H_B = b$ and $H_C = c$ with $a, b, c \in R$, i.e. the classical codes have 1 by 1 (scalar) matrices over $R$, and thus the first sector of Eq.~\eqref{eq:HGP_kunneth} becomes
\begin{equation}
    H_1(\QAB)\otimes H_0(\mc{C}_C) = H_1(\QAB)\otimes R/(c),
\end{equation}
where $(c)$ denotes an ideal of $R$ generated by $c \in R$, i.e. $(c)=\{rc:r\in R\}$. 
% Then we have that $\bar{\gamma}_1$ is injective if and only if $H_1(\QAB)/[(c) H_1(\QAB)] \simeq H_1(\QAB)$. 
According to Ref.~\cite{eberhardt2024logical}, we further have $H_1(\QAB) \simeq R/(a, b) \oplus R/(a, b)$, where $(a, b)$ denotes an ideal generated by $a, b \in R$. This leads to 
\begin{equation}
    H_1(\QAB)\otimes H_0(\mc{C}_C) \simeq R/(a, b, c) \oplus R/(a, b, c).
\end{equation}
Therefore, $\bar{\gamma}_1$ is injective if and only if $c \in (a, b)$ (over $R$).
To this end, we have derived a sufficient algebraic condition for $\bar{\gamma}_1$ to be injective for the (multivariate) tricycle/bicycle codes: 
\begin{enumerate}
    \item The defining group order $|G|$ is odd.
    \item The defining polynomials need to satisfy $c \in (a, b)$.
\end{enumerate}

\subsection{Logical transversality of case studies\label{sec:appx_transversalitycasestudies}}
We now show that, for all three of the examples of 3D tricycle codes in Table I of the main text, we have $c \in (a, b)$, and thereby achieve logical transversality:
\begin{enumerate}
    \item For the $\llbracket 27, 3, 3\rrbracket$ code over $R=\mbb{F}_2[x,y]/(x^3+1,y^3+1)$ with $a = x^2y + x^2y^2$, $b = 1+xy^2$, $c = x+x^2y$. Set $u=x^2,v=x^2\in R$. 
    \begin{align*}
        ua + vb & = 1(x^2y + x^2y^2)+x(1+xy^2)\\
        & = x^2y+x^2y^2 + x+x^2y^2\\
        & = x+x^2y\\
        & = c
    \end{align*}
    \item For the $\llbracket 45, 3, 4\rrbracket$ code over $R=\mbb{F}_2[x,y]/(x^3+1,y^5+1)$ with $a = x+y^2$, $b = 1+xy^2$, $c = x+xy^3$. Set $u=1, v=x^2+xy^3 \in R $. 
    \begin{align*}
        ua+vb & = 1(x+y^2)+(x^2+xy^3)(1+xy^2) \\
        & = x+y^2+x^2+x^3y^2+xy^3+x^2y^5 \\
        & = x+y^2+x^2+ y^2+xy^3+x^2 \\
        & = x+xy^3 \\
        & = c
    \end{align*}
    \item For the $\llbracket 81,3,5\rrbracket$ code over $R=\mbb{F}_2[x,y]/(x^3+1,y^9+1)$ with $a=xy^3 + x^2y$, $b = 1+xy^8$, $c = x^2y^4+x^2y^6$. Set $u=y^3+y^4, v=xy^6+xy^7 \in R$, we have
    \begin{align*}
        ua+vb & = (y^3+y^4)(xy^3 + x^2y)+(xy^6+xy^7)(1+xy^8)\\
        & = xy^6+x^2y^4+xy^7+x^2y^5+xy^6+x^2y^{14}\\
        & \quad \quad +xy^7+x^2y^{15}\\
        & = x^2y^4+x^2y^6\\
        & = c
    \end{align*}
    \item For the $\llbracket 108,6,6\rrbracket$ code from \cite{menon2025magic} over $R = \mbb{F}_2[x,y,z]/(x^3+1,y^3+1,z^4+1)$ with $a = y^2z + x^2yz^3$, $b=x^2 + x^2yz^2$, and $c=z^2 + yz + x + x^2yz^3$. Set $u = z + xz^2 + xyz^3 + xy^2z^3 + x^2z + x^2y$ and $v=z+y\in R$.
    \begin{align*}
        ua+vb & = (z + xz^2 +\cdots + x^2y)(y^2z + x^2yz^3)\\
        & \quad +(z+y)(x^2 + x^2yz^2)\\
        & = y^2z^2+x^2yz^4 +\cdots+yx^2+x^2y^2z^2\\
        & = z^2 + yz + x + x^2yz^3\\
        & = c
    \end{align*}
    \item For the $\llbracket 240,6,8\rrbracket$ code from \cite{menon2025magic} over $R = \mbb{F}_2[x,y,z]/(x^4+1,y^4+1,z^5+1)$ with $a=y^3 + x^2yz^2$, $b=xz^4 + x^3y^3z$, $c=xy^2z^3 + xy^3z^4 + x^2y^2z + x^2y^3z^2$. Set $u=y^2z^4+xz^4+xy^3z^4+x^2z+x^2z^2+x^3yz$, $v=y^2+xy^2z^2\in R$. \begin{align*}
        ua+vb & = (y^2z^4+xz^4+\cdots +x^3yz)(y^3 + x^2yz^2)\\
        & \quad +(y^2+xy^2z^2)(xz^4 + x^3y^3z)\\
        & = y^5z^4+x^2y^3z^6\cdots+x^2y^2z^6+x^4y^5z^3\\
        & = xy^2z^3 + xy^3z^4 + x^2y^2z + x^2y^3z^2\\
        & = c
    \end{align*}
\end{enumerate}

We emphasize that this is only a sufficient condition and often not necessary. For instance, the trivariate tricycle code in Table I of the main text is defined over a group with even order, nevertheless satisfying the logical transversality. 
% Although we couldn't prove a general condition for $\partial^A,\partial^B,\partial^C$ for the induced $\bar{\gamma}_1$ to be injective, it is easy to check whether a given $\bar{\gamma}_1$ matrix is injective computationally, by checking that a basis set of nontrivial Z logical operators in $H_1(\QAB)$ is not mapped to a Z stabilizer in $H_1(\QABC)$.

\section{Details of the codes and physical implementation\label{sec:appx_codedetails}}
Here, we provide in Table~\ref{table:classical_code_choice} the concrete choices of the classical codes for obtaining the quantum codes in Table I of the main text.
\begin{table*}
\centering
% \begin{tabular}{l|c|c|c|c|c}
\begin{tabular}{c|c|c|c|c|c|c}
\toprule
\textbf{Code} & $R$ & $H_A$ & $H_B$ & $H_C$ & \textbf{2D} $\llbracket n,k,d\rrbracket$ & \textbf{3D} $\llbracket n,k,d\rrbracket $ \\
\hline
Pentagon & $\mbb{F}_2$ & $\begin{pmatrix}
    1 & 1 & 0  \\
    0 & 1 & 1 \\
    1 & 0 & 1\\
\end{pmatrix}$ & {\tiny   \setlength{\arraycolsep}{1pt}$\begin{pmatrix}
    1 & 0 & 0 & 0 & 1 & 1 & 0 & 0 & 1 & 0\\
    1 & 1 & 0 & 0 & 0 & 0 & 1 & 0 & 0 & 1\\
    0 & 1 & 1 & 0 & 0 & 1 & 0 & 1 & 0 & 0\\
    0 & 0 & 1 & 1 & 0 & 0 & 1 & 0 & 1 & 0\\
    0 & 0 & 0 & 1 & 1 & 0 & 0 & 1 & 0 & 1
\end{pmatrix}$} & $\begin{pmatrix}
    1 & 1 & 0  \\
    0 & 1 & 1 \\
    1 & 0 & 1\\
\end{pmatrix}$ & $\llbracket 45, 7, 3\rrbracket $ & $\llbracket 180, 8, 3\rrbracket $ \\
\hline
Lifted toric & $\mbb{F}_2[x]/(x^2 + 1)$ & $\begin{pmatrix}
    1 & x  \\
    1 & 1  \\
\end{pmatrix}$ & $\begin{pmatrix}
    1 & x  \\
    1 & 1  \\
\end{pmatrix}$ & $\begin{pmatrix}
    1 & x  \\
    1 & 1  \\
\end{pmatrix}$ & $\llbracket 16, 2, 4\rrbracket $ & $\llbracket 48, 3, 4\rrbracket $ \\
\hline
Lifted toric & $\mbb{F}_2[x]/(x^2 + 1)$ & $\begin{pmatrix}
    1 & x & 0 \\
    0 & 1 & 1 \\
    1 & 0 & 1
\end{pmatrix}$ & $\begin{pmatrix}
    1 & x & 0 \\
    0 & 1 & 1 \\
    1 & 0 & 1
\end{pmatrix}$ & $\begin{pmatrix}
    1 & x & 0 \\
    0 & 1 & 1 \\
    1 & 0 & 1
\end{pmatrix}$ & $\llbracket 36, 2, 6\rrbracket $ & $\llbracket 162, 3, 6\rrbracket $ \\
\hline
BB/BT & {\scriptsize$\mbb{F}_2[x,y]/(x^3 + 1, y^3 + 1)$} & {\scriptsize$a = x^2y + x^2y^2$} & {\scriptsize$b = 1+xy^2$} & {\scriptsize$c = x+x^2y$} & $\llbracket 18, 2, 3\rrbracket $ & $\llbracket 27, 3, 3\rrbracket $ \\
% BB/BT & $\mbb{F}_2[x,y]/(x^3 + 1, y^3 + 1)$ & $a = xy^2 + x^2$ & $b = x^2 + x^2 y$ & $c = y^2 + xy$ & $\llbracket 18, 2, 3\rrbracket $ & $\llbracket 27, 3, 3\rrbracket $  & 3  \\
\hline
BB/BT  & {\scriptsize$\mbb{F}_2[x,y]/(x^3+1,y^5+1)$} &  {\scriptsize$a = x+y^2$} & {\scriptsize$b = 1+xy^2$} & {\scriptsize$c = x+xy^3$} & $\llbracket 30, 2, 5\rrbracket $ & $\llbracket 45, 3, 4\rrbracket $  \\
\hline
BB/BT  & {\scriptsize$\mbb{F}_2[x,y]/(x^3+1,y^9+1)$} & {\scriptsize$a=xy^3 + x^2y$} & {\scriptsize$b = 1+xy^8$} & {\scriptsize$c = x^2y^4+x^2y^6$} & $\llbracket 54, 2, 6\rrbracket $ & $\llbracket 81, 3, 5\rrbracket $ \\
\hline
TB/TT  & {\scriptsize$\mbb{F}_2[x,y,z]/(x^2+1,y^5+1,z^7+1) $} & {\scriptsize$a = y^2z^2 + xy^2z^4$} & {\scriptsize$b = 1 + xy^2z^3$} & {\scriptsize$c = yz + y^4z^2$} & $\llbracket 140, 2, 8\rrbracket $ & $\llbracket 210, 3, 7\rrbracket $ \\
\hline
TB/TT \cite{menon2025magic} & {\scriptsize $\mbb{F}_2[x,y,z]/(x^3+1,y^3+1,z^4+1)$} & {\scriptsize$a=y^2z + x^2yz^3$} & {\scriptsize$b=x^2 + x^2yz^2$} & {\scriptsize$c=z^2 + yz + x + x^2yz^3$} & $\llbracket 72,4,6\rrbracket$ & $\llbracket 108,6,6\rrbracket$ \\
\hline
TB/TT \cite{menon2025magic} & {\scriptsize$\mbb{F}_2[x,y,z]/(x^4+1,y^4+1,z^5+1)$} & {\scriptsize$a=y^3 + x^2yz^2$} & {\scriptsize$b=xz^4 + x^3y^3z$} & {\scriptsize$c=xy^2z^3 + xy^3z^4 + x^2y^2z + x^2y^3z^2$} & $\llbracket 160,4,8\rrbracket$ & $\llbracket240,6,8\rrbracket$\\
\hline
\hline
\end{tabular}
\caption{\textbf{Detailed construction of the 3D/2D codes}. We provide the defining group algebra $R$ and the choices of the classical check matrices $H_A$, $H_B$, $H_C$ over $R$ for each pair of 3D-2D codes in Table I of the main text. For these examples, we construct the 2D code as the product of $H_A$ and $H_B$. The codes in the last two rows are constructed using the same polynomials as the trivariate tricycle codes from Table 2 of \cite{menon2025magic}. 
}
\label{table:classical_code_choice}
\end{table*}

\begin{figure*}
    \includegraphics[width=1\textwidth]{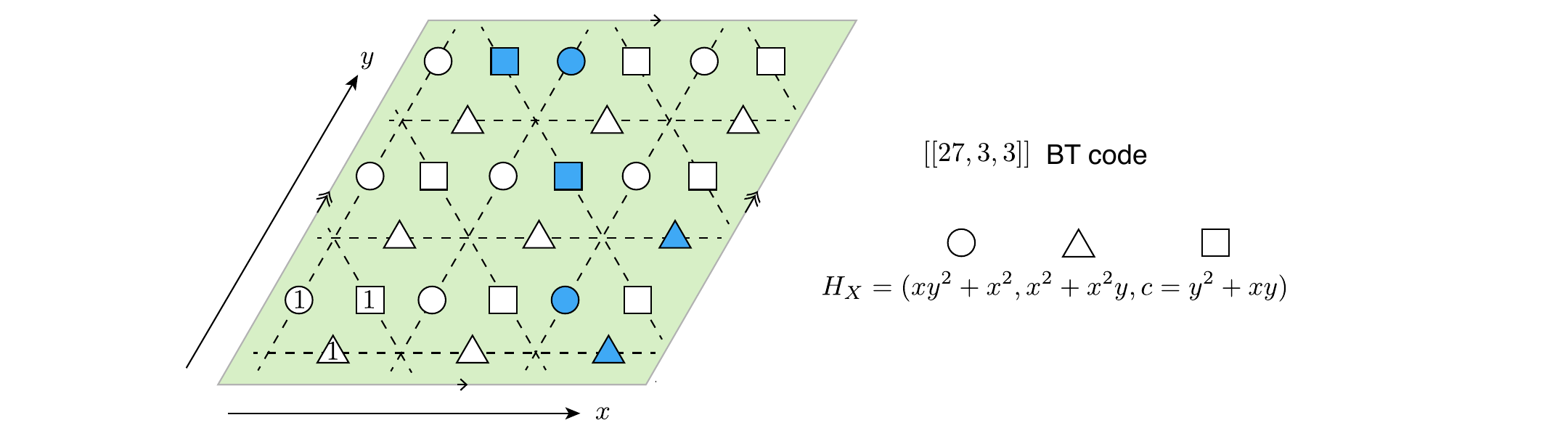}
    \caption{\textbf{Physical implementation of a BT code on a 2D torus with translationally invariant checks}. The three sectors of physical qubits are indicated by the circles, the triangles, and the squares, respectively, embedded on a triangular lattice on a 2D torus with periodical boundary conditions. 
    For a BT code defined over a bivariate polynomial ring $R = \mbb{F}_2[x, y]/(x^{l_x} + 1, y^{l_y} + 1)$, the qubits in each sector are labeled by monomials $\{x^i y^j\}_{i \in [lx], j \in [l_y]}$, with the qubits labeled by $1$ explicitly shown. The support of an example $X$ check is highlighted in blue, consisting of all the monomials involved in the $X$ check matrix $H_x = (a, b, c)$, where $a, b, c$ are the defining polynomials of the code. All other $X$ checks are defined by the translations of this example check on the torus. The $Z$ checks have an analogous translationally-invariant structure (see text). The example code is a $\llbracket 27, 3, 3\rrbracket$ BT code with polynomials $a = xy^2 + x^2$, $b = x^2 + x^2y$, $c = y^2 + xy$ with $l_x = l_y = 3$, which is equivalent to the first BT code in Table~\ref{table:classical_code_choice}.}
    \label{fig:physical_layout}
\end{figure*}

In addition, we show in Fig.~\ref{fig:physical_layout} that the bivariate tricycle (BT) codes in Table~\ref{table:classical_code_choice}, despite their 3D code structure, can be embedded on a 2D torus with translationally invariant checks --- analogous to the layout for the bivariate bicycle (BB) codes~\cite{bravyi_high-threshold_2024}.

Recall that, a 3D BT code has check matrices (see Eq.~\eqref{eq:3D_check_mats}):
\begin{equation}
\begin{aligned}
    H_X & = \begin{bmatrix} 
    a & b & c\\
    \end{bmatrix}, \\
    H_Z & = \begin{bmatrix}
        b^* & a^* & 0\\
        c^* & 0 & a^* \\
        0& c^* & b^*
    \end{bmatrix},
\end{aligned}
\label{eq:3D_BT_codes}
\end{equation}
where $a, b, c \in \mbb{F}_2[x, y]/(x^{l_x} + 1, y^{l_y} + 1)$
where $(\bullet)^*$ denotes the conjugate of each polynomials. 
As illustrated in Fig.~\ref{fig:physical_layout}, when embedding the three blocks of the physical qubits corresponding to the three columns in Eq.~\eqref{eq:3D_BT_codes} on a triangular lattice on a 2D torus, where the qubits in each block are labeled by the monomials $\{x^i y^j\}_{i \in [l_x], j \in [l_y]}$, the $X$ checks $H_X$ are translationally invariant. 
The $Z$ checks can be divided into three groups, corresponding to the three rows of $H_Z$ in Eq.~\eqref{eq:3D_BT_codes}, each being translationally invariant --- analogous to the $X$ checks.
\section{Cup product and Symmetric Triple Cup Product (STCP) CCZ\label{sec:appx_cupproduct}}
In this work we use two related constructions of transversal CCZ gates: the cup product from \cite{breuckmann_cups_2024} and the symmetric triple cup product (STCP) from \cite{menon2025magic}. Both start from the 3D code expressed as a 4-term cochain complex
\[\begin{tikzcd}[cells={nodes={minimum height=2em}}]
C^0 \arrow[r,"d_0"] & C^1 \arrow[r,"d_1"] & C^2 \arrow[r, "d2"] & C^3.
\end{tikzcd}\]
where we identify qubits with 1-cochains $C^1$, X-checks with 0-cochains $C^0$, and Z-checks with 2-cochains $C^2$. A binary cup product is a family of bilinear maps $\cup: C^i\times C^j \rightarrow C^{i+j}$. Taking the product of three classical codes with a cup product, one obtains an induced cup product on the quantum code, which induces a cup product on the cohomologies. This corresponds to a circuit defined by the trilinear function $f_{\cup}(q_i,q_j,q_k)=|(q_i\cup q_j)\cup q_k|\pmod{2}$ where $q_i,q_j,q_k$ indicate qubits from each of the three code blocks, and a physical CCZ gate is applied between those qubits if $f_{\cup} = 1$. The construction of this cup product CCZ relies on a set of conditions on the coboundaries (check matrices) of the classical codes, which ensure that the circuit produced by $f_\cup$ corresponds to a cohomology invariance and thus corresponds to a diagonal logical operation. All codes which we constructed (all except the last two rows of Table \ref{table:classical_code_choice}) used this cup product method. We refer the readers to Refs.~\cite{breuckmann_cups_2024, jacob2025single} for details about this cup-product construction. 

The codes in the last two rows of Table \ref{table:classical_code_choice} from \cite{menon2025magic} are examples of tricycle codes with a constant depth CCZ constructed using the Symmetric Triple Cup Product (STCP) method. It is similar to the above cup product method, but it relaxes the conditions on the classical codes via a modified and more symmetrical trilinear function to achieve higher-rate examples. This is paired with a symmetric integrated Leibniz rule. The corresponding trilinear function similarly defines a code-space preserving CCZ circuit. We refer the readers to Ref.~\cite{menon2025magic} for more details of this symmetrical cup-product construction.

\section{Details for simulations\label{ref:appx_simulationdetails}}
\begin{figure*}
    \centering
    \includegraphics[width=1\textwidth]{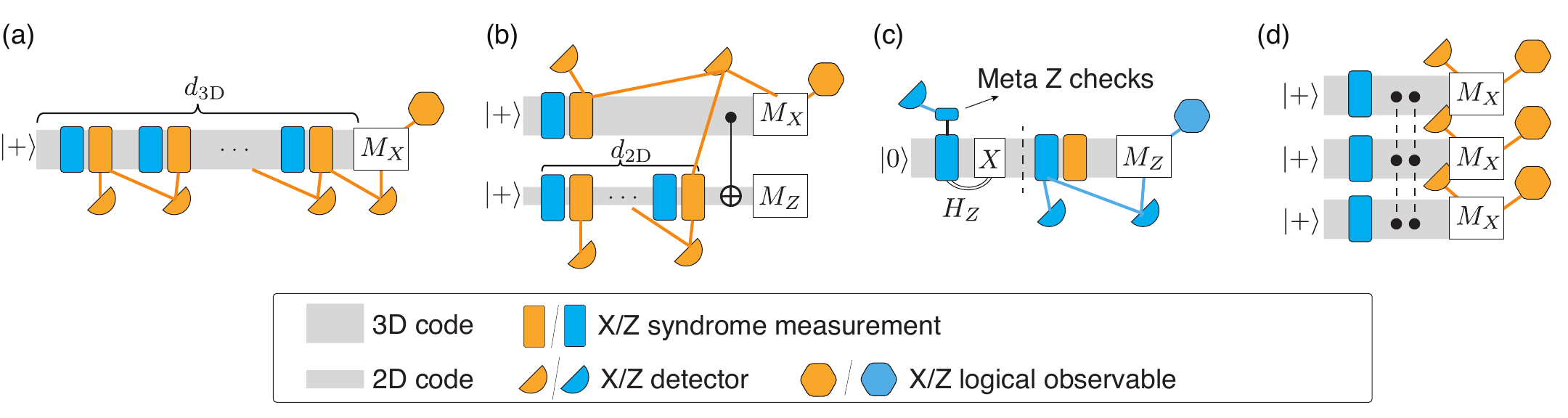}
    \caption{\textbf{Illustration of the numerically simulated circuits.} (a) $X$-basis memory experiment (Fig. 2(a) in the main text). (b) Transversal teleportation from the 2D code to the 3D code (Fig. 2(b) in the main text). (c) Single-shot logical $\ket{+}$ state preparation (Fig. 2(c) in the main text). (d) Magic state preparation (Table II in the main text). }
    \label{fig:simulation_circuits}
\end{figure*}

In this section, we provide the details of the numerical simulations performed in this work. Unless specially noted, the circuits are all Clifford circuits and we simulate them efficiently using Stim~\cite{gidney2021stim}. We consider the standard circuit-level depolarizing noise model, but excluding idling errors. 
For simplicity, we use the coloration syndrome extraction circuit~\cite{tremblay2022constant} that extracts $Z$ and $X$ syndromes separately each in a CNOT depth $w$ for any CSS code with degree-$w$ check matrices.
More compact circuits with interleaving $X$ and $Z$ syndrome extraction in a total CNOT depth $w$ could also be used~\cite{menon2025magic}.
We note that all 3D codes considered in this work have a biased distance $d_X \gg d_Z = d$ and their $Z$ checks are redundant, giving rise to meta $Z$ checks $M_Z$ satisfying $M_Z H_Z = 0$ ($\mod 2$) (see Sec.~\ref{sec:appx_background}).

As illustrated in Fig.~\ref{fig:simulation_circuits}(a), we simulate the memory logical error rates of a $\llbracket n_{\mr{3D}}, k_{\mr{3D}}, d_{\mr{3D}} \rrbracket$ 3D code (Fig. 2(a) of the main text) by initializing the physical qubits in the $\ket{+}^{\otimes n_{\mr{3D}}}$ states, performing $d_{\mr{3D}}$ rounds of syndrome extractions, and transversally measuring the data qubits in the $X$ basis. We set the logical $X$ observables to the $k_{\mr{3D}}$ logical operators reconstructed from the final transversal measurements.
We set $X$ detectors by taking the parities of adjacent-round $X$ syndromes in a standard way, and decode them using the efficent minimum-weight-perfect-matching (MWPM) decoder~\cite{Higgott2025sparseblossom} as the $X$ check matrices $H_X$ have column-weight $2$.
Let $P$ denotes the logical error probability --- the probability that any logical observable takes $-1$ after correction ---  we report the logical error rate per logical qubit per code cycle: $p_L = 1 - (1 - P)^{1/(k_{\mr{3D}} d_{\mr{3D}})}$.
We note that these simulations only probe the logical $Z$ error rates, whereas the logical $X$ error rates are negligible in comparison since $d_X \gg d_Z$.

As illustrated in Fig.~\ref{fig:simulation_circuits}(b), we estimate the logical error rates of the transversal dimension-jump schemes (Fig. 2(b) of the main text) by simulating the transversal teleportation circuit from the 2D codes to the 3D codes (Fig. 1(b) of the main text). Again, we only prove the logical $Z$ error rates by simulating the teleportation of logical $\ket{+}$ states. 
Starting with a 3D code and its $\llbracket n_{\mr{2D}}, k_{\mr{2D}}, d_{\mr{2D}} \rrbracket$ component 2D code initialized transversally in the $X$ basis, we perform $1$ and $d_{\mr{2D}}$ rounds of syndrome extractions for the 3D code and the 2D code, respectively, apply the homomorphic CNOT, and transversally measure the 3D code and the 2D code in the $X$ and $Z$ basis, respectively. 
Note that we include only $1$ code cycle for the 3D code since it supports single-shot preparation of logical $\ket{+}$ states, as will be demonstrated in Fig.~\ref{fig:simulation_circuits}(c).
Also, we omitted the Pauli $X$ feedback corrections on the 3D code since it commutes with the final $X$ basis measurement in this case. We set the logical observables to be the $k_{\mr{3D}}$ logical $X$ operators reconstructed from the final $X$ measurements on the 3D code.
We set the parity $X$ detectors for the 3D codes and 2D codes separately in a standard way, except for the final correlated $X$ detectors on the 3D code whose parities also include the final-round $X$ syndromes on the 2D code before the homomorphic CNOT. 
We decode these $X$ detectors jointly using the BP+OSD decoder\cite{roffe2020decoding}.
Let $P$ denotes the logical error probability, we report the logical error rate per 2D logical qubit: $p_L = 1 - (1 - P)^{1/k_{\mr{2D}}}$.

Next, we demonstrate single-shot state preparation of logical $\ket{+}$ states for the 3D codes (Fig. 2(c) of the main text). A standard protocol achieves this by (1) initializing the physical qubits transversally in $\ket{+}$, (2) measuring one round of $Z$ checks, (3) applying corrections on the $Z$ syndromes using the meta $Z$ checks, (4) applying a Pauli $X$ correction to fix the (corrected) $Z$ syndromes to $+1$. Such a protocol aims for preparing the logical $\ket{+}$ states of the 3D code (with $+1$ stabilizers) up to locally stochastic errors. 
However, if the code does not possess single-shot properties such as soundness~\cite{campbell2019theory} and meta checks, there is no effective mechanism to infer the syndrome errors at step (3) since the $Z$ checks are nondeterministic, and the residual syndrome errors would lead to large residual $X$ data errors after fixing the $Z$ checks to $+1$ at step (4), resulting in states far from the code space (see 2D surface code results in Fig. 2(c) of the main text). 
To capture such potential error mechanisms, we perform the following simulation:
As shown in Fig.~\ref{fig:simulation_circuits}(c), starting with a 3D code initialized transversally in the $Z$ basis, we measure the $Z$ checks once, apply another round of syndrome extraction (including both $X$ and $Z$ checks), and transversally measure in the $Z$ basis. 
Although the initial $Z$ checks are deterministic in this case, we assume they are random and do not utilize them as detectors. Instead, we perform the following two-step procedure to fix the stabilizers and prepare the logical state: 1. perform syndrome correction using the meta check matrix $M_Z$ and a relay-BP decoder~\cite{muller2025improved} 2. based on the corrected $Z$ syndromes, apply a Pauli $X$ correction to fix them to $+1$ using the $Z$ check matrix $H_Z$ and an integer-programming decoder~\cite{cain2024correlated}.
Note that the integer-programming decoder would not be necessary in the true senario for logical $\ket{+}$ state preparation, where the $Z$ syndromes are random and essentially no decoders would be needed (any correction that fix syndromes to $+1$ suffices). Note that in the simulation here, since the $Z$ checks should be deterministically $+1$, any $X$ corrections would represent the residual $X$ errors on top of the prepared logical state. 
We then probe the size of these errors using the subsequent QEC cycle: we perform the standard memory-type decoding using the subsequent $Z$ checks (assuming they are deterministically $+1$) and the final transversal $Z$ measurements, and check logical errors by setting the logical observables to the $Z$ logical operators reconstructed from the final transversal measurements. We decode these detector syndromes using the efficient relay-BP decoder again. Let $P$ denote the logical error probability, we report the logical error rate per logical qubit $p_L = 1 - (1 - P)^{1/k_{\mr{3D}}}$.

Next, as shown in Fig.~\ref{fig:simulation_circuits}(d), we estimate the logical error rate of the magic states prepared via: (1) transversally initialize three 3D code blocks in $\ket{+}$, (2) measure the $Z$ checks once and fix the syndromes (corrected using the meta checks described above) to $+1$, (3) apply the (depth-$2$) CCZ circuit. 
Such a protocol fault-tolerantly prepares the magic states since: (a) based on the simulation in Fig.~\ref{fig:simulation_circuits}(c), we know that the $Z$ checks of the code are fixed to $+1$ up to locally stochastic $X$ residual errors, (b) the $X$ checks are by default $+1$ up to locally stochastic $Z$ errors due to the $X$ basis initialization, (c) the CCZ circuit is of low depth with limited error propagation. 
However, instead of adding the physical CCZ gates, which would make the circuit non-Clifford and hard to simulate, we only add approximate error channels associated with theses gates described in the following.
The residual $X$ errors after the single-shot state preparation will propagate to $CZ$ errors, which we approximate by correlated Pauli $ZZ$ errors; we model the new gate errors associated with the CCZ gates with three-qubit depolarizing noise after each gate. 
We approximate the above errors by adding the following Pauli error channel after the CCZ circuit: let $\delta \in \mbb{F}_2^{n_1 \times n_2 \times n_3}$ represent the physical CCZ circuit across the three code blocks, for each $i \in [n_1]$, $j \in [n_1]$, and $k \in [n_3]$ we add a multi-qubit $Z$ error supported qubits indexed by $\{(i, j^{\prime}, k^{\prime}) \mid \delta_{i, j^{\prime}, k^{\prime}} = 1\}$, $\{(i^{\prime}, j, k^{\prime}) \mid \delta_{i^{\prime}, j, k^{\prime}} = 1\}$, and $\{(i^{\prime}, j^{\prime}, k) \mid \delta_{i^{\prime}, j^{\prime}, k} = 1\}$, respectively, with a probability $p$.
For instance, for the BT codes with a depth-$2$ CCZ circuit, each physical qubit in one code block couples to two qubits in each of the other code blocks, resulting in a collection of five-qubit correlated $Z$ errors.
Finally, we measure the code transversally in the $X$ basis and probe the logical errors by checking the values of the logical $X$ operators. 
We also construct $X$ detectors using these final transversal measurements and perform postselection whenever these detectors are nontrivial. 
Note that this only probes the logical $Z$ errors, and we assume that the logical $X$ error rate is negligible by running very light postselection on the $Z$ detectors using, e.g. soft information~\cite{gidney2024magic}, as these codes have biased distance $d_X \gg d_Z$. We report the logical error rate as simply the logical error probability that any of the $3k_{\mr{3D}}$ logical qubits have an error, which is related to the fidelity of the entangled magic states across three 3D code blocks.
\subsection{Calculation of space-time cost\label{sec:appx_calcspacetimecost}}
Finally, we provide the details on calculating the space-time cost of the magic-state preparation scheme in Table II of the main text.
For the space-time cost values listed in Table II of the main text, we calculate it according to the following formula:
\begin{equation}
    \text{\begin{tabular}{c}
        Space-time cost\\
        per logical qubit
    \end{tabular}}
    = \text{\begin{tabular}{c}
        Physical qubits\\
        per logical
    \end{tabular}}\times\text{rounds}
      \times (r+2),
\end{equation}
where $r$ is the maximum weight of any stabilizer.
% \begin{equation}
%     \text{Space-time cost per logical qubit}=(\text{physical qubits per logical}\times\text{rounds})\times(\text{max weight of any stabilizer}+2).
% \end{equation}

For the toric and BB/BT/TB/TT codes in our construction, we use 
\begin{equation}
    \text{\begin{tabular}{c}
        Physical qubits\\
        per logical
    \end{tabular}}=\frac{(n + \text{total \# of stabilizers})}{k}
\end{equation} and $\text{rounds}=1/(\text{success rate})$.
% \begin{equation}
%     \text{Space-time cost} = \frac{(n + \text{\# of X stabilizers}+\text{\# of Z stabilizers})}{(\text{code distance}\times \text{success rate})}\times(\text{max weight of any stabilizer} + 2)
% \end{equation}
In particular, the top four rows of Table II are calculated as 
% (following the same structure as above):
\begin{equation}
    \begin{aligned}
        \text{BT } \llbracket 27,3,3\rrbracket: \quad & \frac{27 + 9 + 27}{3\cdot 0.704}\times(6+2) = 239\\
    \text{BT } \llbracket 81,3,5\rrbracket: \quad &
\frac{81 + 27 + 81}{3\cdot 0.349}\times(6+2) = 1444\\
\text{Toric } \llbracket 81,3,3\rrbracket: \quad &
\frac{81 + 27 + 81}{3\cdot 0.349}\times(6+2) = 1444 \\
\text{Toric } \llbracket 375,3,5\rrbracket: \quad & 
\frac{375 + 125 + 375}{3\cdot 0.008}\times(6+2) = 291,667
    \end{aligned}
\end{equation}
For the color codes, we obtained the value of $\text{qubits}\times\text{rounds}$ from the Figure 1 of \cite{gidney2024magic} (specifically, the x-axis values of the points labeled ``ungrown d=3" and ``ungrown d=5" for the $\llbracket 7,1,3\rrbracket$ and $\llbracket 19,1,5 \rrbracket$ color codes, respectively).

\section{Exact code distance calculations\label{sec:appx_exactcodedistance}}
For all the codes examples listed in Table~\ref{table:classical_code_choice} -- which include all codes used in our simulations -- we verify the exact code distance using a binary integer linear program (ILP). We compute the distance of each code by formulating the search for minimal-weight logical operators as a series of ILPs, one for each logical qubit. Since these are all CSS codes, we can find a basis of $k$ logical operators, $L_X$ and $L_Z$, in the X and Z bases, respectively. For each logical X-operator $\overline{X}_i$ in $L_X$, introduce binary decision variables $x_j\in\{0,1\}$ representing whether a Pauli Z acts on qubit $j$. The objective is to minimize the Hamming weight $\sum_j x_j$ of this Z-type operator, subject to the constraints that it commutes with all X-type stabilizers ($H_Xx = 0\pmod{2}$) and anticommutes with the chosen X-logical ($x^T\overline{X}_i=1\pmod{2}$). Let $w_i$ denote the optimal objective value for logical operator $\overline{X}_i$. Then, Z-distance is $d_Z = \min(w_1,\ldots,w_k)$. We obtain the X-distance $d_X$ analogously by exchanging $H_X$ and $H_Z$ and working with logical Z-operators. The above method, implemented using Gurobi \cite{gurobi}, quickly converged for all the code examples except for X-distance of the trivariate tricycle code $\llbracket 210,3,7\rrbracket$ which is higher than the Z-distance.

\end{appendix}

% \bibliography{refs}

\end{document}